\pdfoutput=1
\documentclass[aps,pre,superscriptaddress,reprint,notitlepage,twocolumns]{revtex4-1}
\usepackage{natbib}
\usepackage{graphicx}
\usepackage{ifpdf}
\usepackage[table]{xcolor}
\usepackage{array}
\usepackage{amssymb}
\usepackage{caption}
\usepackage{amsmath}
\usepackage{enumerate}
\usepackage{subfig}
\usepackage{float}
\usepackage{amsfonts}
\usepackage[hyperindex,breaklinks]{hyperref}
\usepackage[all]{hypcap}
\usepackage{amsthm}
\usepackage{cleveref}
\newtheorem{myprop}{Proposition}

\newtheorem{mylemma}{Lemma}
\newtheorem*{mythm}{Theorem}

\begin{document}
\title{An effective network reduction approach to find the dynamical repertoire of discrete dynamic networks}

\author{Jorge G. T. Za\~nudo}\email[]{jgtz@phys.psu.edu}
\affiliation{Department of Physics, The Pennsylvania State University,\\
  University Park, Pennsylvania, 16802-6300, USA.}
\author{R\'eka Albert}\email[]{ralbert@phys.psu.edu}
\affiliation{Department of Physics, The Pennsylvania State University,\\
  University Park, Pennsylvania, United States of America.}
\affiliation{Department of Biology, The Pennsylvania State University,\\
  University Park, Pennsylvania, 16802-5301, USA.}

\begin{abstract}
Discrete dynamic models are a powerful tool for the understanding and modeling of large biological networks. Although a lot of progress has been made in developing analysis tools for these models, there is still a need to find approaches that can directly relate the network structure to its dynamics. Of special interest is identifying the stable patterns of activity, i.e., the attractors of the system. This is a problem for large networks, because the state space of the system increases exponentially with network size. In this work we present a novel network reduction approach that is based on finding network motifs that stabilize in a fixed state. Notably, we use a topological criterion to identify these motifs. Specifically, we find certain types of strongly connected components in a suitably expanded representation of the network. To test our method we apply it to a dynamic network model for a type of cytotoxic T cell cancer and to an ensemble of random Boolean networks of size up to 200. Our results show that our method goes beyond reducing the network and in most cases can actually predict the dynamical repertoire of the nodes (fixed states or oscillations) in the attractors of the system.

\end{abstract}
\maketitle

\section*{Lead paragraph} \label{sec:lead}

There is a great interest in understanding how the complex cellular behaviors in living organisms emerge from the underlying network of molecular interactions. Discrete dynamic models, a modeling paradigm in which the dynamical variables can only take discrete states, have been increasingly used to model systems with a large number of components. The feature that makes discrete dynamic models an attractive choice is their ability to reproduce the qualitative dynamics of the system using only the activating or inhibiting nature of the interactions; the knowledge of the rates of the biochemical processes involved is not required. Despite their simplicity, the main impediment in using discrete dynamic models for modeling large systems is combinatorial complexity. In this work, we offer a solution to this problem by introducing a novel network reduction approach. Our reduction approach uses a topological criterion in an augmented representation of the network to identify network components that take a fixed state, which can then be used to shrink the effective network size. A noteworthy virtue of our method is that it can be applied to large network sizes (up to size 200 and beyond). We have found that our method goes beyond reducing the size of the network and can predict the dynamical repertoire of the nodes (fixed states or oscillations).

\section{Introduction} \label{sec:1}

The interactions among cellular elements such as proteins, mRNA, and small molecules are orchestrated in such a way that they support all the complicated behaviors cells are capable of (such as homeostasis, growth, movement, cell differentiation and cell division) \cite{BarabasiNetworkBiology}. In order to get a full understanding of the relation between cellular behaviors and their underlying network of interactions, the construction of informative dynamic models based on the current biological knowledge is very important. Several dynamical modeling techniques exist, which provide different levels of detail in the dynamics, while in turn requiring varying amounts of biological information \cite{LessIsMore,QuantitativeModeling}. At one end of the spectrum, for example, highly quantitative information can be obtained from ordinary differential equation models \cite{TysonDynamics1,TysonDynamics2,QuantitativeReview,SignalingDynamics} by providing different reaction rates (e.g. transcription/translation rates, association/dissociation constants, degradation coefficients) and the biophysical/biochemical properties of the components (e.g. enzyme cooperativity). At the other end, the qualitative dynamics of the system can be reproduced by a discrete dynamic model \cite{RekaDrosophila,ArabidopsisFlower,YeastSocolar,TLGLPNAS,TcellSignaling,RekaReview}, which requires only the combinatorial activating or inhibiting nature of the interactions, and not the kinetic details \cite{RekaNoKineticDetails}.

Given the surprising but demonstrated fact that the essential dynamical properties of a variety of systems can be reproduced without knowing the values of the specific kinetic parameters of the processes involved \cite{RekaDrosophila,ArabidopsisFlower,YeastSocolar,TLGLPNAS,TcellSignaling,RekaReview}, one may wonder if there is a model-independent way to infer the dynamical properties of cellular networks just by using the network topology (graph structure), that is, the identity of the components and knowledge about their interactions. Historically, this relation between structure and dynamics was recognized early on in the pioneering work of Jacob and Monod \cite{JacobMonod}, Thomas \cite{ThomasOriginal}, Kauffman \cite{KauffmanOriginal}, and Glass \cite{GlassKauffman}, and was part of the original motivation for the study of discrete dynamic models. The common idea is that the presence of feedback loops is necessary for the emergence of complex dynamical properties such as multistability and oscillations. More specifically, by assigning a sign to the interactions (+ if activating and - if inhibitory) and to the feedback loops in the network (the sign of a loop is given by the product of the signs of its edges), the following two simple rules were proposed by R. Thomas \cite{ThomasConjectures} to relate the network structure to its dynamics:
\begin{enumerate}
  \item A necessary condition for multistability (multiple stable steady states) is the existence of a positive feedback loop.
  \item A necessary condition for sustained oscillations (limit cycles) is the existence of a negative feedback loop.
\end{enumerate}
Since these early works, extensive research has been done in this direction and the validity of these rules has been demonstrated both in the differential \cite{Plahte,Snoussi,Gouze,Soule} and discrete frameworks \cite{Aracena,RemyRuet,Remy}. Recent works have even extended these rules to include not only necessary but also sufficient conditions for multistability and oscillations \cite{MinimalCircuitsRemy,Siebert}.

Despite all this progress, there is still a need for developing new analysis tools that relate the network structure to its dynamics, especially ones that are applicable to large scale networks. This is a problem for many of the methods developed so far, since many of them are very computationally demanding and can only be exactly applied to networks of small to moderate size. The size of the networks is also a problem even in cases where mathematical theorems are available, because as the network increases in size, it is very likely that the conditions needed in the theorems become harder and harder to be fulfilled. These limitations call for methods that are as generally applicable as possible.

The novel analysis method we present in this work has the objective of inferring the dynamical repertoire of a network based purely on network topology and the combinatorial nature of the interactions. Framed in the discrete dynamic framework, our method is based on the idea that some groups of nodes in the network can only stabilize in a single or a small number of fixed states. By expanding the network to explicitly include the nature of the interactions (positive or negative) and the potentially synergistic regulation of every element in the network, we can identify these stable groups of nodes and use them to simplify the network. The result is a complete reduction (which directly gives the fixed points of the system) or a very simplified network in which most nodes are expected to oscillate. In section \ref{sec:2} we explain our method in more detail, including the network expansion and network reduction techniques involved in it.

\section{Predicting the stable dynamical repertoire of a Boolean model of a biological network} \label{sec:2}

\subsection{Biological networks and discrete dynamic models} \label{sec:2.1}

A network of cellular components can be represented by a directed graph $G=(V, E)$, where $V=\left(v_1, v_2, \ldots, v_N \right)$ are the nodes describing the elements of the system, and $E$ are the edges denoting the directed interactions among the components. To each edge one also commonly associates a sign, which denotes the regulatory nature of the interaction (+ if activating, and - if inhibitory). Although the sign of the interaction enriches the biological information included in the network, experience shows that knowing the nature of the interactions is not enough and the combined effect of the interactions on each element also needs to be considered. For this purpose, every node $v_i$ is assigned a function $f_i$ which depends on the $k_i$ regulators of $v_i$ (and sometimes on itself) and that incorporates the combinatorial nature of the interactions. The set of functions $F=\left(f_1, f_2 , \ldots, f_N \right)$ contains all the dynamical information of the system, and thus, depends on the type of dynamic model used.

For our study, we choose the simplest kind of discrete dynamic model, namely, the Boolean framework, in which the functions in $F$ are taken as logical (Boolean) functions and each node $v_i$ can take one of two possible states: ON (or 1) and OFF (or 0). The biological interpretation of each state varies depending on the context, although in most cases it refers to above (ON) or below (OFF) a certain threshold level. The state of the system at any time can then be denoted by a vector whose $i^{th}$ component is the state of node $v_i$ at that time. As a consequence of the discrete number of states, the state space of the system is finite ($2^N$ states), and a temporal sequence of states can be represented as a trajectory in state space. Every temporal trajectory will eventually reach a set of network states in which it settles down, known as an attractor. An attractor can either be composed of a single network state in which the system stays fixed, known as a fixed point or a steady state, or a group of network states between which it alternates, usually referred to as a complex attractor or an oscillation. In a steady state the state of all nodes remains fixed, while in a complex attractor all or a subset of the nodes keep changing their states (i.e. they oscillate),  and the state of the rest of the nodes (if any) is fixed.

The possible trajectories and complex attractors of a system depend not only on the functions in $F$ but also on the representation of time as a continuous or discrete variable (the fixed points are time-implementation-invariant). The most common choice for Boolean dynamics is taking time as a discrete variable, in which case the nodes are updated at discrete time steps according to the functions in $F$. In the simplest case, the synchronous scheme, every node is updated simultaneously in discrete time steps and its state depends only on the state of the system in the previous step. The synchronous updating scheme, although suitable in some situations, is not apt for our purposes, as it inherently assumes that all processes occur at a similar timescale, which is clearly not true for intracellular networks, in which a large variety of cellular components and processes are involved. Furthermore, the timescales of a large number of these processes are not well known, and even when they are, they may be subject to fluctuations due to cell-to-cell variability and environmental perturbations. Previous work has developed various asynchronous updating methods, wherein each node is updated according to its own timescale, and which can be either deterministic (the timescales are fixed during a simulation) \cite{AsynchronousMethods,Chaves1} or stochastic (the timescales are randomly varied during a simulation) \cite{GlassAsynchronous,ThomasReview,AsynchronousMethods,Chaves1,Chaves2,Socolar}. In a comparative study \cite{AssiehJTB} several discrete time asynchronous updating schemes were tested in the same biological network, with its results suggesting that the general asynchronous method, in which at every discrete time step a randomly selected node is updated, is the most appropriate scheme for our requirements.

The general asynchronous method is advantageous not only because it takes into consideration the multiple timescales involved in intracellular processes, our incomplete knowledge of all these timescales, and the inherent stochasticity of biological processes, but also because it allows a natural biological interpretation of the attractors of the system. The reason for this is that, by definition, the general asynchronous method samples all possible timescales of the system, therefore, the attractors must correspond to the \emph{patterns of activity of the system which are invariant with respect to arbitrary fluctuations in the rates of the processes involved}, which we denote the \emph{stable dynamic repertoire of the network}. The use of the Boolean framework and the general asynchronous updating scheme then maps the problem of finding the rate-invariant dynamic behavior of a cellular network into finding the attractors of a Boolean network.

\subsection{Finding the attractors of a Boolean model} \label{sec:2.2}

As it was pointed out in section \ref{sec:2.1}, the state space of a Boolean network with $N$ nodes contains $2^N$ states. This exponential dependence on the number of nodes of the state space's size makes the problem of finding the attractors of a Boolean network computationally intractable, which means a full search of the state space can be performed, in practice, only for small networks ($N \lesssim 20$). To overcome this challenge several types of methods have been proposed to simplify the search space. The most prominent of these approaches, the so-called network reduction methods \cite{DecimationProcess,AssiehJTB,ReductionNadil,ReductionVeliz}, shrink the effective network size by removing frozen nodes (nodes that reach the same steady state regardless of initial conditions) \cite{DecimationProcess,SocolarandKauffman,DrosselScaling} and dynamically irrelevant nodes (such as simple mediator nodes or nodes with no outputs), and by simplifying the Boolean functions (for example, due to the presence of a node with a fixed state), thus reducing the effective state space.

Although the reduction methods developed so far have been successfully applied to several biological networks \cite{AssiehJTB,AssiehPCB,ReductionNadil,ReductionVeliz}, they have some limitations. For example, some of them may not be effective enough, in the sense that even after reduction the state space's size is still unmanageable, in which case one must resort to sampling the state space \cite{AssiehJTB,AssiehPCB}. Other methods can affect the state space in such a way that the attractor space is changed and only some types of attractors (e.g. steady states) are preserved \cite{ReductionNadil,ReductionVeliz}.

\subsection{The role of stable motifs and oscillating components in the attractor landscape} \label{sec:2.3}

The method that we propose is based on the idea that certain components of the network can only stabilize in one or a small number of attractors. This idea is itself not new and is closely related to R. Thomas' rules linking feedback loops with the appearance of complex dynamical behavior in biological networks \cite{ThomasConjectures}. For example, an approach that connects the dynamics of certain network motifs to construct the attractors of the full Boolean network was recently proposed by Siebert \cite{Siebert}. The novelty of our method is the efficiency of identifying every motif that stabilizes in an asynchronous attractor despite being coupled to the rest of the network. The main insight of our approach is that a representation of a Boolean network known as the expanded network can be used to easily identify these network components and their states. By combining the knowledge of the behavior of this group of nodes with network reduction methods, we can find other network components that stabilize as a consequence of the former. By repeatedly applying this procedure, the states of all nodes can be found, which will correspond to their states in the attractors of the system.

We will first focus on finding the components that stabilize in a fixed state. More specifically, we will look for network components with certain topological characteristics that cause themselves and other nodes to take a fixed state. We refer to these network components as \emph{stable motifs} or \emph{stable components}. It is worth noting that the nodes in the stable motifs may or may not be part of the so-called frozen nodes of Boolean networks \cite{DecimationProcess,SocolarandKauffman,DrosselScaling}, since the nodes of these stable motifs can have more than one steady state. Although it may seem at first that this restriction to fixed-state nodes reduces the general applicability of the method, it turns out that this is not necessarily the case. Indeed, there are two possibilities after using our method to find all the nodes that take a fixed state in an attractor. The list of fixed-state nodes could include all nodes in the network, in which case we have identified a fixed point attractor. Or, if the list covers a subset of the nodes, these nodes must represent the non-oscillating nodes of a complex attractor. In this last case, the nodes the method does not identify are expected to keep changing their values (i.e. oscillate) in the attractor. We refer to the network components that stabilize in an oscillating state as \emph{oscillating motifs} or \emph{oscillating components}. We discuss in more detail the role of these oscillating components and of oscillations in section \ref{sec:2.7}.

\subsection{Expanded Network} \label{sec:2.4}

In order to identify the stable motifs of a Boolean network, it is convenient to use a representation that incorporates explicitly the update functions $f_i$. Previous work \cite{ESMRuiSheng} has shown that a useful representation for this purpose is the so-called expanded network representation.

\begin{figure*}[t]
\includegraphics{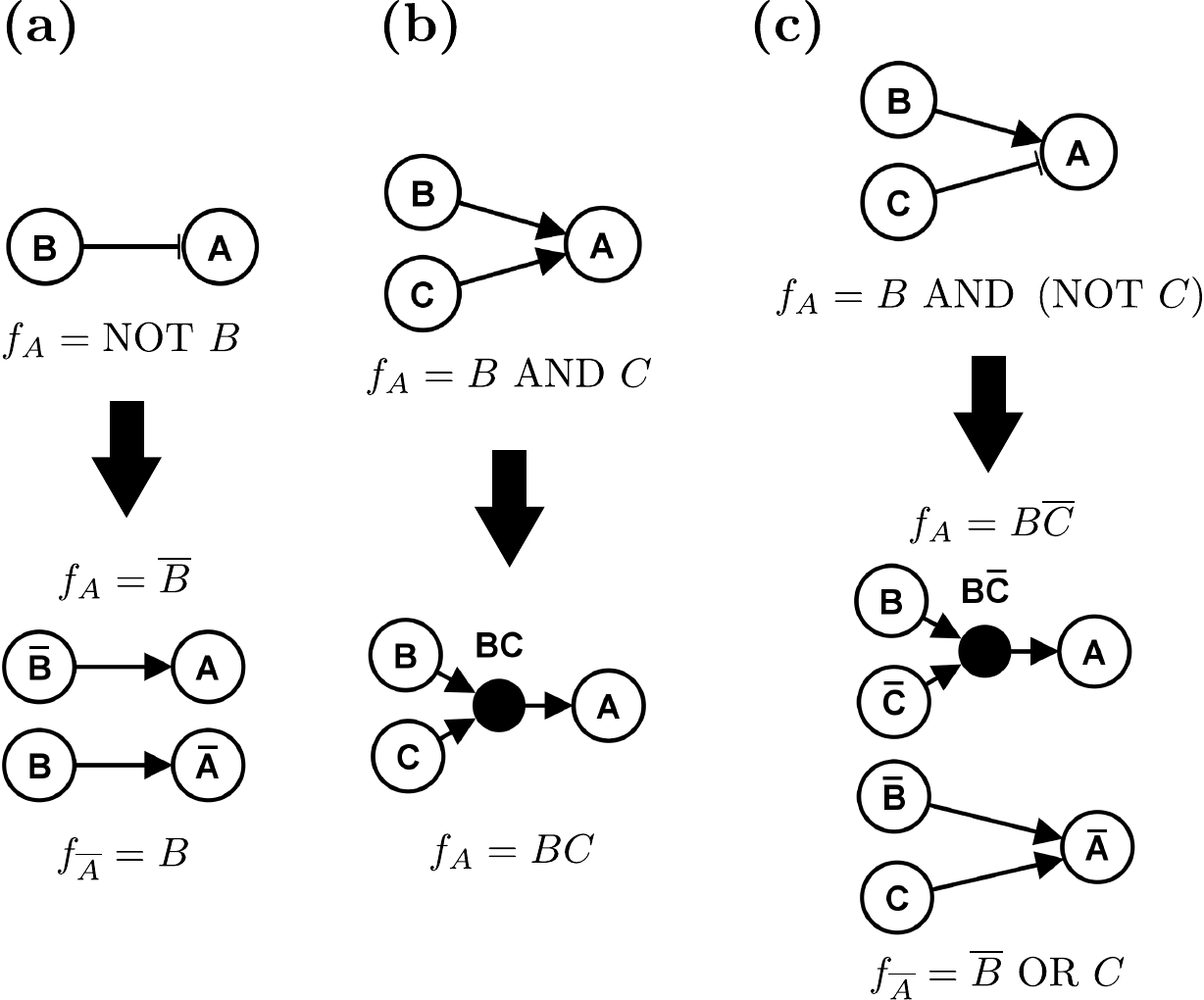}
\caption{Operations for the creation of the expanded network. (a) Node $A$ has the update function $f_A=\hbox{NOT } B$. The addition of complementary nodes introduces a new node $\overline{A}$ with update function $f_{\overline{A}}=\hbox{ NOT } f_A = B$. It also introduces a complementary node for node $B$, which makes the update function of $A$ take the form $f_A=\hbox{NOT } B=\overline{B}$. (b) Node $A$ has the update function $f_A=B\hbox{ AND }C$. The addition of composite nodes introduces a new node $BC$ that represents the cooperative effect of B and C on A. (c) Node $A$ has the update function $f_A=B\hbox{ AND }\left(\hbox{NOT }C\right)$. The two expansion operations introduce complementary nodes for $A$, $B$, and $C$, and a composite node for the AND relation between B and $\overline C$.}
\label{fig:NetworkExpansion}
\end{figure*}

The creation of the expanded network consists of two basic operations, which we illustrate in Figure \ref{fig:NetworkExpansion}. First, in networks that include account negative regulations, we introduce a complementary node $\overline{v}_i$ for every node $v_i$ in the network and assign to each $\overline{v}_i$ an update function which is the Boolean negation of $v_i$'s update function $f_i$. The addition of complementary nodes has a two-fold effect; not only does it allow us to evaluate the inhibitory effect of a node on the rest of the network, but by assigning the negation of the original update function to every complementary node, it also explicitly considers how the other nodes can promote the inactivation of a given node. Note that because of the addition of complementary nodes all the edges are of the same (positive) nature, and thus, no sign needs to be specified.

As an example, let us consider node $A$ with update function $f_A=\hbox{NOT } B$, as illustrated in Figure \ref{fig:NetworkExpansion}(a), where for simplicity we denote the state of the node with the node name. The addition of complementary nodes means that we add a new node $\overline{A}$ with update function $f_{\overline{A}}=\hbox{ NOT } f_A = B$. An additional complementary node is added for node $B$, which makes the update function of $A$ take the form $f_A=\hbox{NOT } B=\overline{B}$. The expanded network contains two positive edges, from $\overline{B}$ to $A$, and from $B$ to $\overline{A}$, instead of the negative edge from $B$ to $A$.

Second, to incorporate the combinatorial nature of the update functions, we introduce a composite node for each set of synergistic interactions (that is, AND relationships) in the Boolean functions $f_i$. For example, consider the case shown in Figure \ref{fig:NetworkExpansion}(b), in which node $A$ has the logical function $f_A=B\hbox{ AND }C$. Since the function contains an AND relationship between node $B$ and node $C$, a composite node $BC$ is added when expanding the network. Node $B$ and $C$ are connected by directed edges to the composite node $BC$, and an edge from $BC$ to $A$ is also added. A more complicated example in which both operations are applied is shown in Figure \ref{fig:NetworkExpansion}(c).

In general the introduction of composite nodes may not be as obvious as shown in these examples if we have nontrivial combination of AND, NOT and OR rules in the functions $f_i$. Because of this, it is convenient to represent each update function $f_i$ with $K$ input nodes $\left\{ v_{i1}, v_{i2}, \ldots , v_{iK}  \right\}$ in the following disjunctive normal form:
\begin{align*}
  f_i =& \left(s_{1} \hbox{ AND } s_{2} \hbox{ AND } \cdots \hbox{ AND } s_{k} \right) \\
  & \hbox{ OR }\left(s_{k+1} \hbox{ AND } s_{k+2} \hbox{ AND } \cdots s_{l} \right) \\
  & \hbox{ OR } \cdots  \hbox{ OR } \left(s_{m} \hbox{ AND } s_{m+1} \hbox{ AND } \cdots \hbox{ AND } s_{n} \right),
\end{align*}
where the $s_{j}$'s are either the states of one of the $K$ input nodes of $f_i$, or one of these states' negations. In the same way, the negation of the update function, $\overline{f}_i$, is also represented in a disjunctive normal form. Once the functions are represented in the disjunctive normal form, the introduction of composite nodes is simple: a composite node is added for every set of nodes involved in a conjunctive clause (AND-dependent relationship). In the following we will refer to nodes that are not complementary nor composite as normal nodes.

\subsection{Identifying Stable Motifs from the Expanded Network} \label{sec:2.5}
\addcontentsline{toc}{subsection}{Identifying Stable Motifs from the Expanded Network}

We define a stable motif in the expanded network as any of the smallest strongly connected components (SCCs) in the expanded network representation which satisfy these two properties: (1) the SCC does not contain both a node and its complementary node, and (2) if the SCC contains a composite node, all of its input nodes must also be part of the SCC. The first condition makes sure that there is no contradiction between the SCCs found and a state in the original Boolean network, wherein every node can either take the value 0 (which would correspond to having the complementary node in the SCC) or 1 (which would correspond to having the normal node in the SCC). The second condition is a consequence of the synergistic nature of composite nodes, which means that a composite node and all of its inputs form an irreducible unit. By smallest SCC we refer to any SCC that does not contain another SCC with the two specified properties, but that, otherwise, is arbitrary in size. For example, a node of the expanded network with a self-loop would be one of these smallest SCCs. We restrict ourselves to the smallest SCCs so that when there are multiple such SCCs in the system all possible combinations of these SCCs are considered as steady states of the network. Note that this does not result in the loss of information on larger possible SCCs; the remaining parts of these SCCs will be found after the smallest SCCs are reduced.

The composition of the stable motif directly determines the state of a subset of nodes in the Boolean network: every normal node of the stable motif will adopt the state 1, and for every complementary nodes included in the stable motif the corresponding node of the Boolean network will adopt the state 0.

\begin{figure*}[t]
\includegraphics{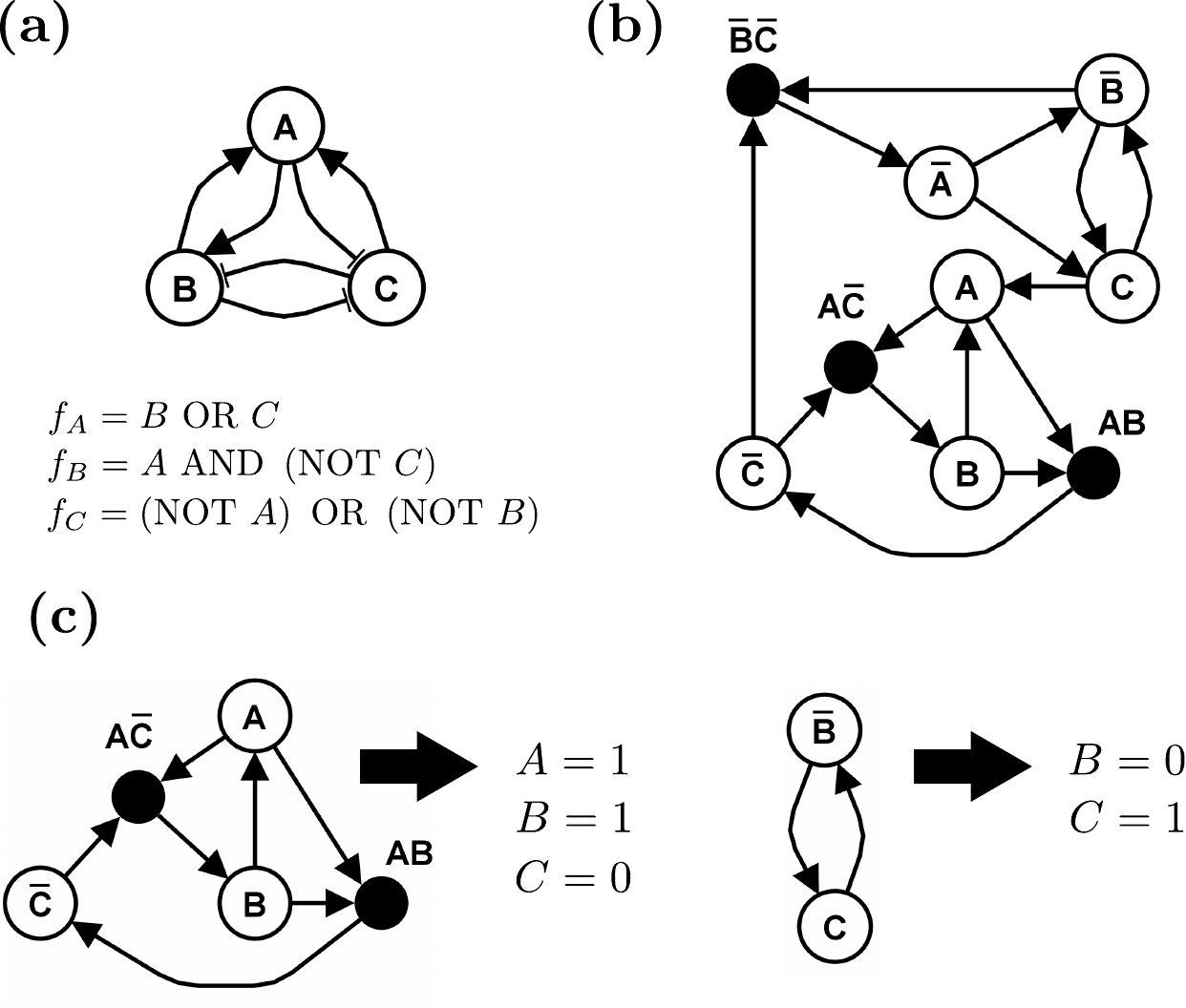}
\caption{Identification of stable motifs from the expanded network. (a) An example of a Boolean network. (b) The expanded network representation of the Boolean network in (a). (c) The two stable motifs in the expanded network, that is, the two smallest SCCs in the network that satisfy the requirements of not containing both a node and its complementary node, and containing all the inputs of every included composite node. Each stable motif indicates the fixed states of the corresponding subset of nodes of the Boolean network.}
\label{fig:MotifIdentification}
\end{figure*}

As an example, consider the Boolean network and its expanded network representation in Figure \ref{fig:MotifIdentification}. The smallest SCCs that satisfy both stable motif requirements are $\left\{A,B,\overline C\right\}$ and $\left\{\overline B, C\right\}$, both of which are shown in Figure \ref{fig:MotifIdentification}(c). The corresponding states for these stable motifs are $\left\{A=1,B=1,C=0\right\}$ and $\left\{B=0,C=1\right\}$, respectively.

So far, we have only defined a stable motif in terms of the expanded network representation. We can extend the concept of stable motif to the original Boolean network to mean the nodes in the Boolean network whose state is specified by a stable motif of the expanded network. These nodes include all the normal nodes that are included in the stable motif of the expanded network and all the normal nodes whose complementary nodes are included in the stable motif of the expanded network. In this way, a stable motif can mean either a set of nodes in the expanded network that satisfy the two requirements or the set of nodes in the Boolean network whose state is specified by a stable motif in the expanded network, depending on the context.

It is important to point out that the stable motifs depend on the structure of the logical rules and thus on the topology of the network being considered. This means that an arbitrary change in the logical rules or in the topology of the network can modify the stable motifs, and we know of no obvious way to determine how the motifs will change without having to reconstruct the expanded network.

\subsection{Network Reduction} \label{sec:2.6}

Once the stable motifs of the network have been identified, the next step is to determine the influence of these nodes on the rest of the network. More specifically, for each stable motif found, we want to find the nodes in the network whose state stabilizes due to the influence of this stable component. We adapt the method previously developed by Saadatpour et al. to simplify the network \cite{AssiehJTB,AssiehPCB}, which has been shown to preserve both the fixed points \cite{ReductionNadil,ReductionVeliz} and the complex attractors of the system \cite{AssiehMath}. This method removes not only the frozen nodes of network \cite{DecimationProcess,SocolarandKauffman,DrosselScaling}, but also the nodes that reach a steady state under the influence of a given combination of source node states. It consists of two main steps:

\begin{enumerate}
\item Identify the nodes whose state is fixed during the dynamics, which we will refer to as
    source nodes; for our case these will initially correspond to the nodes in the stable motif being considered. Modify the Boolean functions of the nodes downstream of the source nodes by setting the state of the source node to its fixed value. If a downstream node's modified function can only have one possible outcome, then this node can be used as a source node itself.
\item Remove mediator nodes (i.e., nodes that have only one incoming edge and one outgoing edge) and irrelevant sink nodes (i.e., nodes that have no outgoing edges). For the case of mediator nodes connect the input of the mediator node to its output. The value of the removed nodes will be determined once the value of their input nodes is known.
\end{enumerate}

For each separate stable motif found in the expanded network, these two steps are repeated recursively until neither of them can be applied anymore.

After network reduction, we obtain a set of states for each stable component, each of which corresponds to the states of the nodes in the stable motif and the states of other nodes which stabilized as a consequence of the stable motif. For each of these sets of states, there is also a reduced network that contains the nodes whose state we still do not know. On each of these reduced networks the whole method will be applied again, starting with the creation of the expanded network (section \ref{sec:2.4}) and ending with the network reduction process, and iteratively until there are no more nodes with unknown states or no new stable motifs are found. For the case where there are no more nodes with unknown states, a fixed point attractor of the system is obtained directly from the state of the nodes of the stable components.

For the cases in which there are no new stable motifs in the final reduced networks, the state of the nodes making up said networks is still unknown. Since our method is based on identifying nodes that stabilize in a specific steady state, the expectation is that these leftover nodes will oscillate in an attractor of the system, while in that same attractor the rest of the nodes will take the steady state value found during the simplification process that leads to the reduced network in consideration. For conciseness we will refer to the final output of our method, consisting of a set of stabilized nodes (and their states) and a (potentially empty) set of nodes with undetermined states as a \emph{quasi-attractor}. Our notion of quasi-attractor is closely related to similar other concepts in the literature, such as the singular steady state originally introduced by Snoussi and Thomas \cite{SnoussiThomas,Siebert} and the logical steady state used by Klamt et al. \cite{Klamt}.

Quasi-attractors are closely related to the attractors of a network, both fixed points and oscillations. For example, if the set of oscillating nodes in a quasi-attractor is empty, then the states of the stabilized nodes will correspond to the node states in a fixed point attractor, thus, this quasi-attractor is in fact a fixed point. More generally, for every attractor of the system there exists a quasi-attractor associated to it; this quasi-attractor is such that every node whose state is fixed in the quasi-attractor will also have its state fixed in the same value in the attractor it is associated to. The proof of this is statement is given in Appendix \hyperref[AppendixProof]{A}.

\subsection{Oscillations and oscillating components} \label{sec:2.7}

Using the expanded network representation on networks that show oscillatory behavior, we have found that it can also be used to identify the oscillating components of a network. To find the oscillating components using the expanded network representation, we search for the largest SCCs that satisfy these properties: (1) the SCC must contain the complementary node of every normal node and vice versa, and (2) if the SCC contains a composite node, all its input nodes must also be part of the SCC. The first of these conditions makes sure that all nodes oscillate, by having both states of every node as part of the SCC. The second condition is a consequence of a composite node and all of its inputs forming an irreducible unit. In this case we look for the largest SCCs because we want to find all the nodes that feed back to each other in the oscillation.

These properties are necessary but not sufficient conditions for a group of nodes to oscillate. We have found that there is a third condition that, if also satisfied, is sufficient (though not necessary) for a group of nodes to oscillate, which is that (3) the oscillating component cannot contain stable motifs composed only of normal and complementary nodes. This extra condition is related to the possibility of the coexistence of a steady state and a complex attractor in the sub-state-space. The simplest example that shows this kind of behavior, which we denote \emph{unstable oscillation}, is shown in Figure \ref{fig:Unstableoscillations}. In general, during the reduction process, we need to find the components that could have unstable oscillations (that is, that satisfy the necessary conditions (1) and (2), but not the sufficient condition (3)) to make sure that we preserve all attractors. As a consequence, we obtain a group of quasi-attractors that may not have a corresponding attractor; we refer to these quasi-attractors as marked quasi-attractors in the step-by-step network reduction algorithm in Appendix \hyperref[AppendixAlgorithm]{B}.

\begin{figure*}
\includegraphics{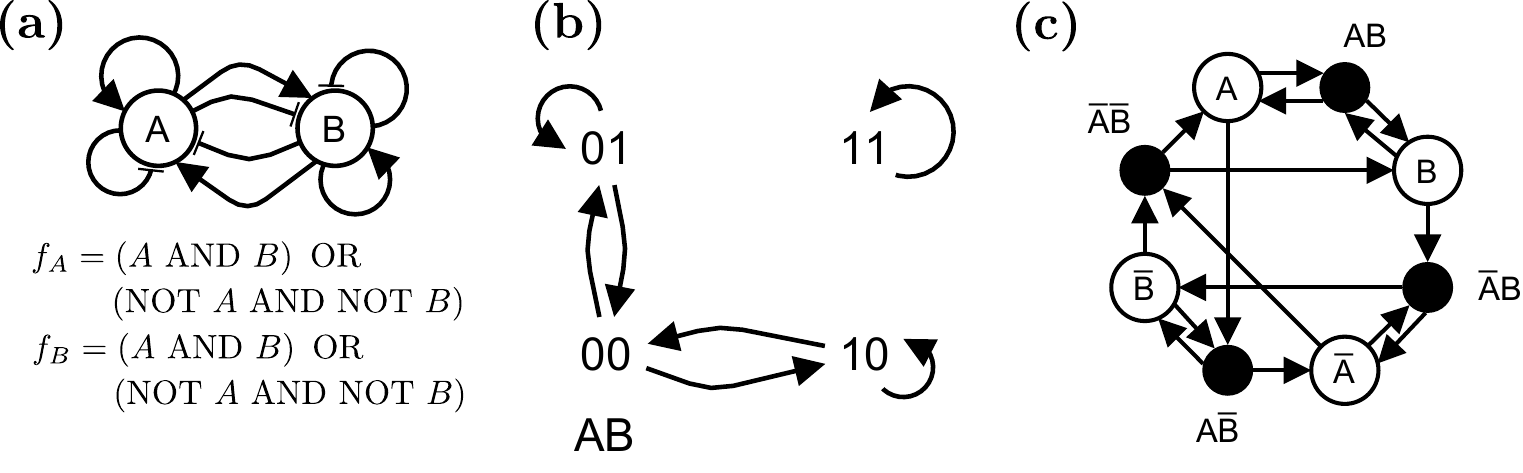}
\caption{An example of a component that has an unstable oscillation. This network has an attractor in which all the nodes oscillate and also has a steady state attractor. (a) The network and its respective Boolean rules. (b) The state transition graph of the network. The nodes of the state transition graph are the states of the system (written in the order $A$,$B$) and the edges represent allowed state transitions when only one node is updated. State 11 is a fixed point as there are no transitions going out of it. States 01, 00, and 10 form a complex attractor. (c) The expanded representation of the network. Note that $\left\{A, B, AB\right\}$ forms a stable motif and that the whole expanded network forms an oscillating SCC.}
\label{fig:Unstableoscillations}
\end{figure*}

Another type of dynamical behavior of the oscillating components that needs to be considered is when the nodes of the oscillating components do not visit all possible states of their sub-state-space in an attractor, which we refer to as an \emph{incomplete oscillation}. Incomplete oscillations are important because a node that is downstream of an oscillating component that displays incomplete oscillations may reach a steady state as a consequence of the nodes of the component only visiting part of their sub-state-space. This type of behavior has been found before in studies of synchronous networks (for example, see Figure 1 in the work by Bilke and Sjunnesson \cite{DecimationProcess}).

As an example of an incomplete oscillation, consider the network shown in Figure \ref{fig:RNoscillations}(a). In this example the nodes $A$ and $B$ oscillate and their state transition graph is shown in Figure \ref{fig:RNoscillations}(b). From the state transition graph one can clearly see the complex attractor is $A,B=\left\{(1,0),(0,0),(0,1)\right\}$. Since once the nodes $A$ and $B$ settle down in the attractor, the state $A=1$, $B=1$ cannot be reached, then node $C$ (whose update function is $f_C=A\hbox{ AND }B$) will necessarily stabilize in the state $C=0$. Note that if either $A$ and $B$ took all possible states in the attractor, or if the update function of $C$ was different, $C$ would also oscillate in the attractor.

\begin{figure*}
\includegraphics{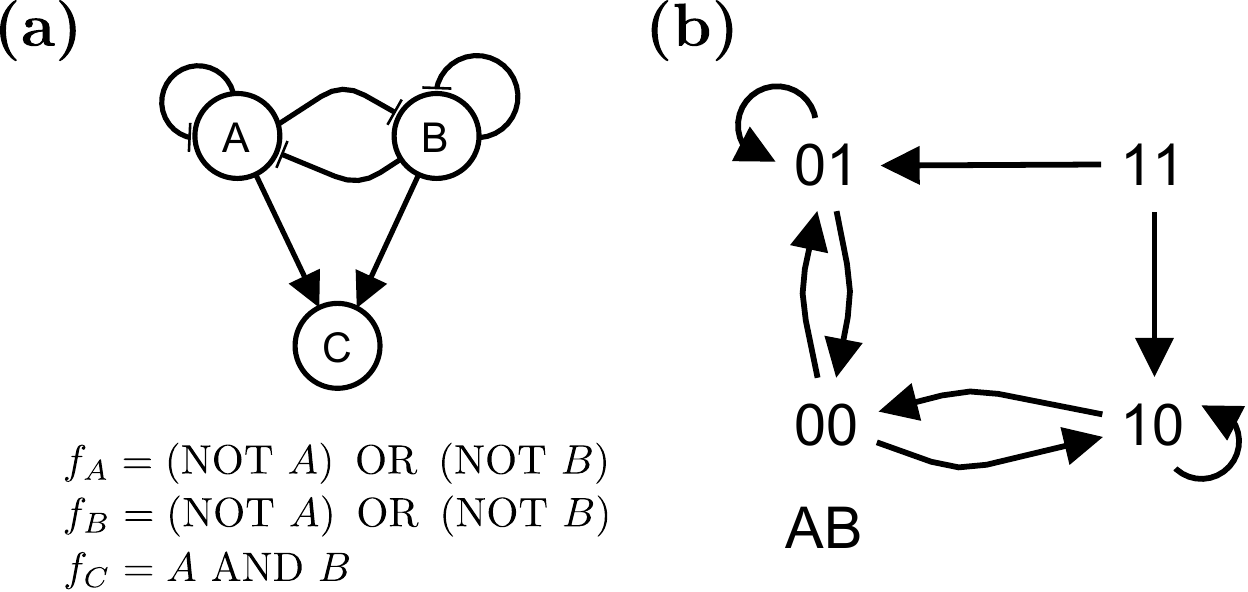}
\caption{An example of a node configuration in which a node can stabilize without the influence of an input signal or a stable motif. In this example $A$ and $B$ oscillate in a complex attractor, but they do not take all possible states of their state transition graph in this attractor. Specifically they miss the $A=1$, $B=1$ state. As a consequence node $C$ stabilizes in the state $C=0$. (a) The node configuration and their respective Boolean rules. (b) The state transition graph of nodes $A$ and $B$. States 01, 00, and 10 form a complex attractor.}
\label{fig:RNoscillations}
\end{figure*}

\section{Results} \label{sec:3}

We implement the network reduction process with a custom Java code. The steps of the network reduction algorithm are described in Appendix \hyperref[AppendixAlgorithm]{B}. The main challenge in implementing the reduction method computationally lies in finding the stable motifs from the expanded network representation. The reason for this is that stable motifs have to be the \emph{smallest SCCs} that satisfy the properties we outlined above, which means that in order to identify all possible stable motifs we need to find all directed cycles that do not contain both a node and its complementary node, since each of them could potentially be the smallest SCC we are searching for. The issue with finding all possible directed cycles is that the time complexity is $O\left((N+E)(c+1)\right)$ (using Johnson's algorithm \cite{CycleAlgorithm}), where $N$ is the number of nodes, $E$ is the number of edges, and $c$ is the number of directed cycles, the latter of which can grow faster than $2^N$ for the worst case scenario of a fully connected network.

Because of the caveats discussed in section \ref{sec:2.7} involving unstable and incomplete oscillations, one may be concerned that other similar cases are not taken into account and that, as a consequence, some attractors could be lost during the reduction process. In Appendix \hyperref[AppendixProof]{A} we address this concern by formally proving that for every attractor of the Boolean network there is a corresponding quasi-attractor that will be found by our reduction method. In order to further test the validity and generality of our network simplification method, we apply it to a previously developed genetic network, and also to an ensemble of random Boolean networks \cite{KauffmanOriginal,KadanoffReview}. For the case of the genetic network, we choose the signaling and regulatory network involved in a type of white blood cell cancer (T cell large granular lymphocyte leukemia or T-LGL leukemia) \cite{TLGLPNAS,AssiehPCB}, while for the ensemble of random Boolean networks we choose the original Kauffman or $N-K$ model \cite{KauffmanOriginal,KadanoffReview}.

\subsection{T Cell Large Granular Lymphocyte Leukemia Network} \label{sec:3.1}

\begin{figure*}
\includegraphics{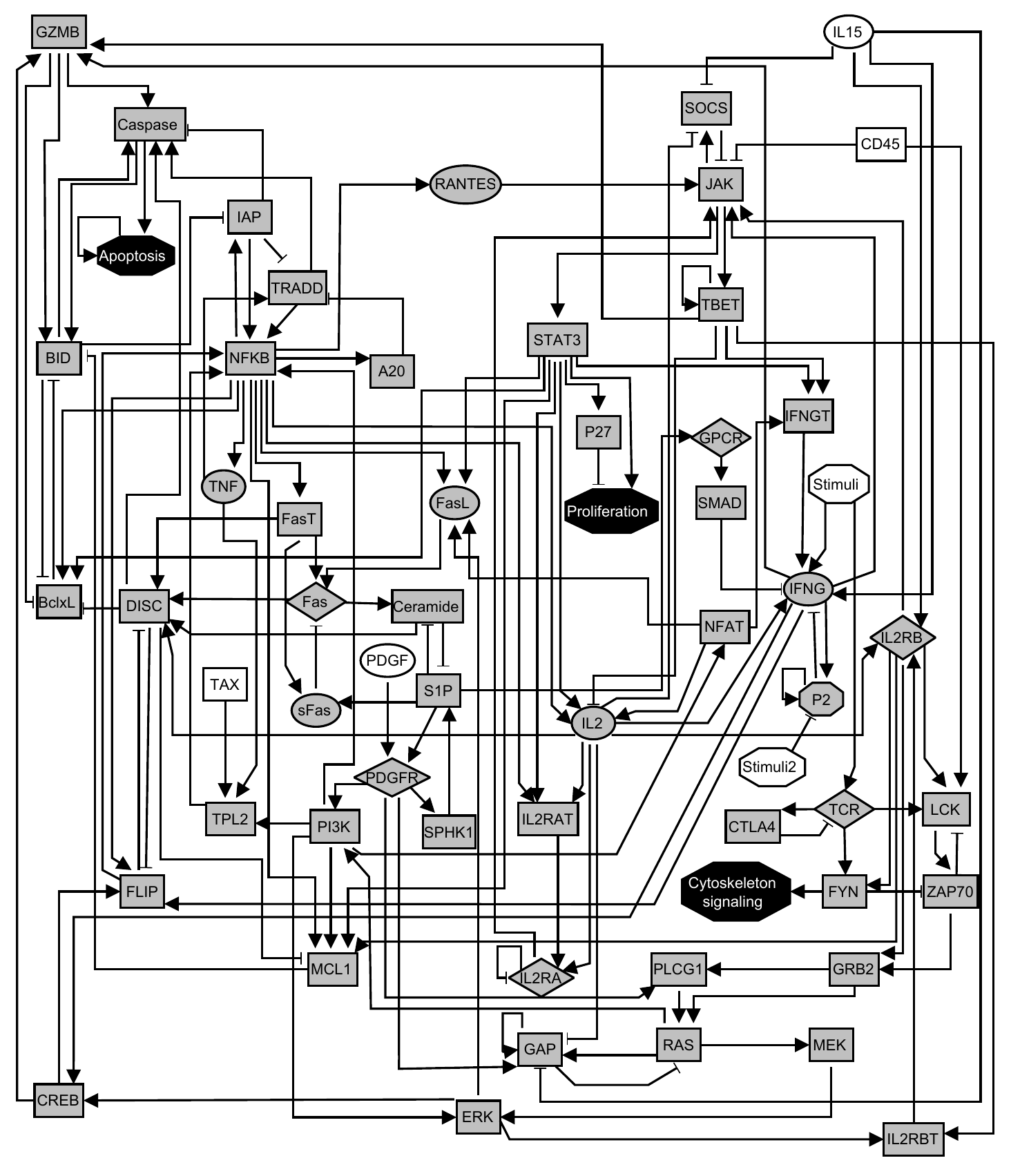}
\caption{The T-LGL leukemia survival signaling network. The shape of the nodes indicates the cellular location or the type of nodes: rectangles indicate intracellular components, ellipses indicate extracellular components, diamonds indicate receptors, and hexagons represent conceptual nodes (Stimuli, Stimuli2, P2, Cytoskeleton signaling, Proliferation, and Apoptosis). Node colors are used to distinguish input nodes (white), output nodes (black) and the rest of the nodes in the network (gray). An arrowhead or a short perpendicular bar at the end of an edge indicates activation or inhibition, respectively. This figure and its caption are adapted from \cite{AssiehPCB}.}
\label{fig:TLGLnetwork}
\end{figure*}

Cytotoxic T cells play a central role in the immune response. When an infection occurs, these T cells detect antigens in infected cells and, in response, trigger a set of intracellular signaling cascades, which lead to the production of cytokines (small signaling molecules) that induce the self-destruction of the infected cells. Normal cytotoxic T cells undergo activation-induced cell death (or apoptosis) after successfully fighting infection, however, in T-cell large granular lymphocyte (T-LGL) leukemia mature cytotoxic T cells survive and, in time, cause an autoimmune disorder. In addition to their abnormal survival, these T cells also show a deregulated activity (higher or lower than in normal T cells) of many signaling pathways and genes.

A Boolean network model of T cell survival in the context of T-LGL leukemia was constructed by Zhang et al. \cite{TLGLPNAS} through an extensive literature search. The logical rules were constructed  such  that the known experimental results in healthy and leukemic cytotoxic T cells were reproduced by the model. These rules, reproduced in Appendix \hyperref[AppendixC]{C}, in many cases do not have a simple form. The resulting network consists of 60 nodes and 142 regulatory edges, with the nodes representing genes, proteins, receptors, or small molecules (Figure \ref{fig:TLGLnetwork}). The network contains 6 nodes with no upstream components which represent external signals (Stimuli, IL15, PDGF, Stimuli2, CD45, and TAX), and also contains 3 output nodes that serve as indicators of biological functions or cell fate (Cytoskeleton signaling, Proliferation and Apoptosis). Two of these input and output nodes play an especially important biological role: Stimuli, which represents antigen stimulation, and Apoptosis, which denotes programmed cell death.

Zhang et al. \cite{TLGLPNAS} used asynchronous Boolean dynamics to show that in the sustained presence of PDGF and IL15 the system may converge to a state that recapitulates all dysregulations in T-LGL leukemia, in addition to the expected state of self-programmed cell death (apoptosis). Later Saadatpour et al. \cite{AssiehJTB} used network reduction to show that, under the presence of said signals, the two attractors found by Zhang et al. (apoptosis and T-LGL leukemia) are the only possible ones. Although in the case studied by Saadatpour et al. the network reduction method was enough to simplify the network to a manageable size (6 nodes), this is actually not the case if one is interested in studying all possible combinations of the input signals, since in many cases the network obtained after reduction is still quite large (30-40 nodes). This then gives an opportunity to apply our reduction method to cases in which previous methods fall short.

We apply our reduction method to all combinations of external signals in the presence of antigen (Stimuli=ON). To validate the quasi-attractors found through our method, we compare them with the attractors obtained by randomly sampling a large number of initial conditions and evolving them using a general asynchronous updating scheme wherein one node is updated at each time step. For all the cases we find that the attractors/quasi-attractors obtained are exactly the same, which together with the proof in Appendix \hyperref[AppendixProof]{A}, shows that the reduction method can indeed be used to find all possible attractors. A table containing all the leukemic attractors is included in Appendix \hyperref[AppendixD]{D} (the Apoptosis=ON attractor, in which all nodes except Apoptosis are inactive, is always a possibility, thus for simplicity we do not include it in the table).

As an example, consider the case in which the IL15 signal is present and the rest of them are not (IL15=Stimuli=ON, PDGF=Stimuli2=CD45=TAX=OFF). Simplifying the network using only the effect of these input signals we obtain a Boolean network of 42 nodes. The expanded network representation of this network has 144 nodes (42 normal nodes, 42 complementary nodes, and 60 composite nodes) and 302 edges. Searching the expanded network, we find four stable motifs from the expanded network representation, which correspond to the states: i) $\left\{\right.$PDGFR = S1P = SPHK1 = ON, Ceramide=OFF$\left.\right\}$, ii) $\left\{\right.$PDGFR = S1P = SPHK1 = OFF$\left.\right\}$ , iii) $\left\{\right.$TBET = ON$\left.\right\}$, and iv) $\left\{\right.$P2 = ON$\left.\right\}$. These motifs are shown in Fig. \ref{fig:TLGLMotifs}. Performing network reduction using each of these four stable motifs leads to reduced networks of widely varying sizes: 3, 39, 35, and 39 nodes, respectively. The reduced network due to the first stable motif consists of three disconnected nodes with self-loops (one negative and two positive ones), and can gives rise both to apoptosis or the T-LGL leukemia attractor. For the networks corresponding to the three remaining motifs we need to continue the reduction process and search for stable motifs in each of these networks.

\begin{figure*}
\includegraphics{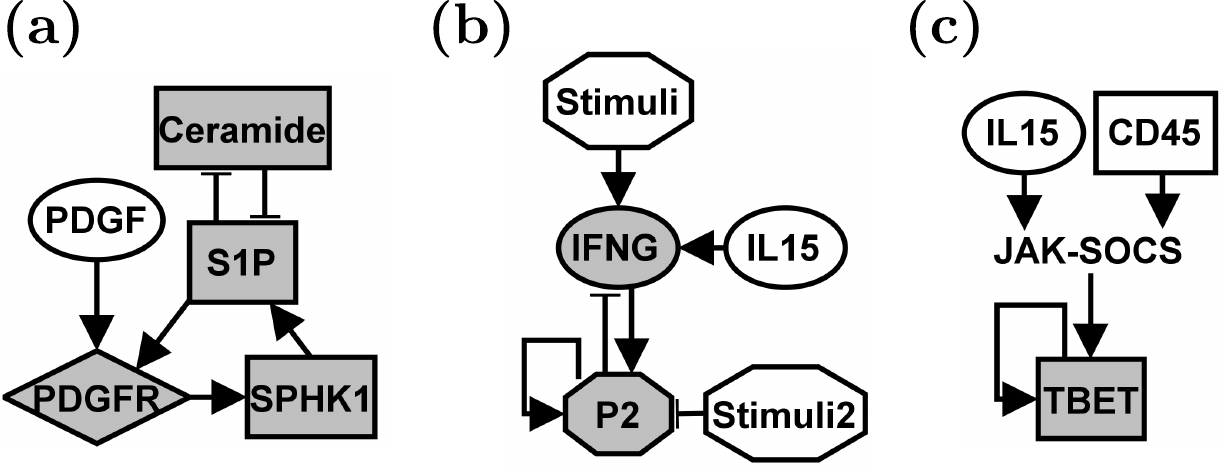}
\caption{The three stable motifs of the T-LGL leukemia network found most often during the reduction process. The actual motifs found and the states in which each of these motifs can stabilize vary depending on the active signals. We also show the input signals (white nodes) that affect these motifs directly or almost directly (for the motif in (c)). (a) The PDGFR-S1P-SPHK1-Ceramide motif, which represents the ceramide/sphingomyelin pathway and the platelet derived growth factor receptor. (b) The IFNG-P2 motif, which is related to the control of the cytokine interferon gamma in CTLs. (c) The TBET motif, which represents the regulation of the T-box transcription factor.}
\label{fig:TLGLMotifs}
\end{figure*}

The stable motifs we find during reduction of the remaining networks are likely to include some of the same four stable motifs found previously (as long as that specific motif was not used to obtain the specific network in consideration), but may also contain stable motifs different from those previously found. For example, the network due to the second motif has two stable motifs, both of which were found in the previous network (iii and iv). The network obtained from the third motif has four stable motifs, one of which is different from the motifs previously found (with states $\left\{\right.$MEK = ERK = RAS = PCLG1 = IL2RBT = IL2RB = GRB2 = ON$\left.\right\}$), and three of which are the same as previous motifs (i,ii, and iv). For the network corresponding to the fourth motif we find three stable motifs, all of which had already been found (i, ii, and iii). If we continue the reduction process we find that the network due to the second motif gives rise to the apoptosis attractor after 2-3 more network reductions (depending on the stable motifs used for the reduction), while the third and fourth motifs can produce both the T-LGL leukemia or apoptosis attractor after 2-4 more network reductions.

\subsection{Ensemble of Random Boolean Networks} \label{sec:3.2}

Random Boolean networks were first introduced by S. Kauffman as a model to understand the general dynamical properties of gene regulation and cell differentiation processes \cite{KauffmanOriginal}, and have been extensively studied ever since \cite{KadanoffReview}. In addition to the original Kauffman networks, several variants of random Boolean networks that could be considered more biologically realistic have also been introduced (for example, models with arbitrary degree distributions \cite{MaxScaleFree}, canalizing functions \cite{Canalyzing}, threshold functions \cite{ZanudoRTN}, or multiple discrete states \cite{Multilevel}). All of these models share the distinguishing feature of the original model, that is, the existence of three dynamical regimes: i) an ordered one in which similar initial conditions typically converge after a transient time, ii) a disordered regime in which the system becomes very sensitive to small changes in the initial conditions, and iii) a critical regime, poised at the boundary of the ordered and disordered regimes, in which perturbations retain their size. Evidence suggests that the gene regulatory networks of living organisms operate near the critical regime \cite{Critical1,Critical2}.

For simplicity we use the original Kauffman or $N-K$ model to test our network reduction method in randomly constructed networks. The $N-K$ model consists of an ensemble of Boolean networks with $N$ nodes in which every element has $K$ input nodes. To construct one of the networks in this ensemble, the $K$ input nodes of every element are chosen randomly from the rest of the network. Every node is then assigned one of the $2^K$ possible Boolean functions at random. To use what is considered the most biologically realistic case of this model, we use different network sizes with degree $K=2$, which is the case at which this ensemble operates in the critical regime \cite{KadanoffReview}. For this case it has been shown that the number of relevant nodes increases as $N^{1/3}$, and that the number of asynchronous attractors grows as a power law in $N$ \cite{DrosselAsynchronous}; one would then expect the existence of an efficient method to find this relatively small number of attractors with a relatively small fraction of relevant nodes.

To test the validity of our network reduction method, we compare the final reduced networks obtained with the asynchronous attractors of the original network for an ensemble of Kauffman networks of different sizes ($N=$ 5, 10, 15, 18, 25, 50, 100, 150 and 200, with an ensemble size of $\Omega=200$ for $N\leqslant18$ and $\Omega=100$ for $N\geqslant25$). For networks up to size $N=100$ we were able to use the exact Johnson's algorithm to find the networks' cycles, however, for larger networks we needed to restrict the search to cycles of less than $40$ nodes. To compare the system's attractors with the result of our reduction method (the quasi-attractors) we focus only on the nodes whose state stabilizes; if for every quasi-attractor there is one attractor that exactly matches the stabilized states of the quasi-attractor, we then say the results are compatible. Note that since our reduction method does not predict the actual state of the nodes remaining after reduction of the stable motifs, it is not necessary for these node states to agree. The na\"ive expectation is that these remaining nodes oscillate in the attractors. If this is indeed the case, we then say that the results are equivalent.

To find the attractors for small networks ($N\leq18$), we construct the asynchronous state transition graph, a directed graph on the unit $N$-cube whose nodes represent the states of the system and whose edges are the allowed transitions between states as a result of a single node's update \cite{GlassAsynchronous,ThomasReview}. In the asynchronous state transition graph  attractors correspond to sink SCCs, that is, SCCs of states whose outgoing edges can only lead to other states of this same SCC. More specifically, fixed points correspond to single states with self-edges and no other outgoing edges, while complex attractors correspond to sink SCCs made up of more than one state.

For large networks ($N\geq25$) we cannot construct the whole asynchronous state transition graph because of its enormous size, so we resort to sampling the state space to look for attractors. In particular we use a method based on the one by Wang et al. \cite{RuiShengAsynch}. In this method we construct part of the state space by starting from a large number $N_S$ of initial conditions and following the system's trajectory for $T$ effective time steps, that is, we make sure that at every step one node changes value (unless a fixed point is reached, in which case no nodes can change value). To find the attractors from the resulting partial state transition graph we use the same criteria as in the complete state transition graph. To avoid false positives, we check the validity of every attractor obtained with this method by starting from one of the states in the putative attractor, updating it $T_{transient}$ effective time steps, then creating a partial state transition graph with $T_{search}$ effective time steps, and finally searching for attractors in the resulting state transition graph. For our case we use $N_S=5000$, $T=300$, $T_{transient}=1500$ and $T_{search}=50000$.

\begin{figure}
\includegraphics{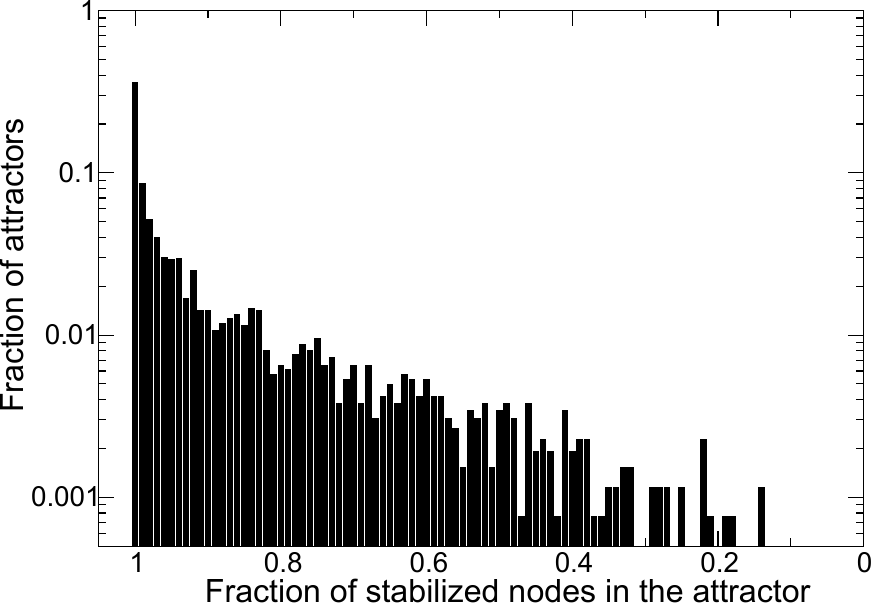}
\caption{Distribution function for the fraction of stabilized nodes in an attractor for $N=100$ for an ensemble of $\Omega=1000$ networks. Note the logarithmic scale in the vertical axis.}
\label{fig:StabilizedNodes}
\end{figure}

Of the total 1300 networks that were compared we find that in all but five networks (all of which had $N\geq100$) the results of our method and of the attractor identification/sampling methods were equivalent. For the remaining five networks we find that the results were compatible, that is, although the state of the nodes predicted to have stabilized by the reduction process matched in the results of both methods, there were some nodes that did not stabilize according to the reduction method that were found to take a fixed state in the attractors found by sampling (i.e. using the partial state transition graph). We reiterate that this disagreement does not mean that our method is incorrect; our reduction method does not actually predict the state of the nodes remaining in the final reduced networks. It is also worth pointing out that in most quasi-attractors the fraction of nodes that don't stabilize is relatively small, as shown in Figure \ref{fig:StabilizedNodes}. Although these remaining nodes are expected to oscillate in the attractor, this is actually not necessary, as discussed in section \ref{sec:2.7}.

\begin{figure}
\includegraphics{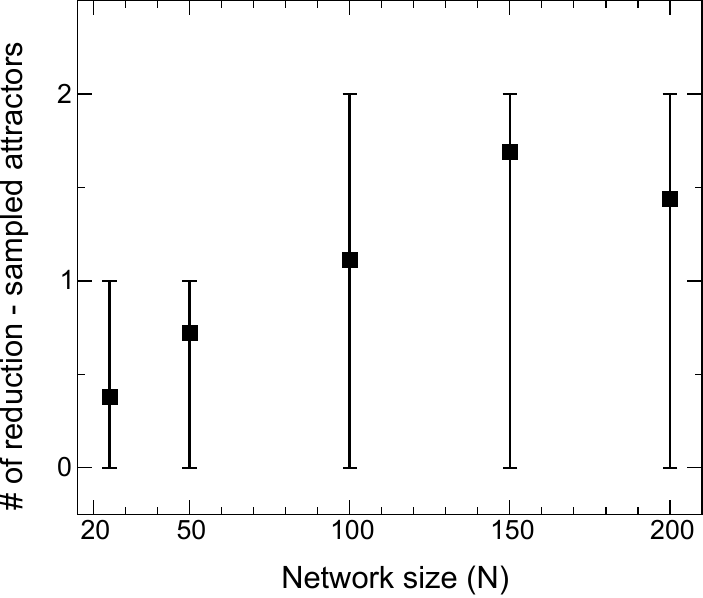}
\caption{Difference in the number of attractors found between the reduction and sampling methods. The squares represent the average difference in the number of attractors between the two methods, while the lower and higher limits of the bars represent the 20th and the 80th percentile of the distribution of attractor number difference. In all the cases the difference is zero or positive, that is, the reduction method never finds less attractors than the sampling method. For all network sizes shown an ensemble size of $100$ networks was used.}
\label{fig:RNattractors}
\end{figure}

We also compare the number of attractors found by each of the three methods. For small networks we find that network reduction and the exact method always find the same number of attractors/quasi-attractors. For large networks we find that the reduction method always finds either more or the same number of attractors as the sampling method. In Figure \ref{fig:RNattractors} we show the average, the 20th percentile, and the 80th percentile of the difference in the number of attractors found by the two methods for large networks. To make sure that the quasi-attractors found by reduction are real attractors, we check their validity by constructing a partial state transition graph  just as we did for the sampling method. In some cases, in which the putative attractor was expected to have a large number of oscillating nodes ($\gtrsim$20), an attractor could not always be found with this method. For these cases we analyzed the trajectory in the partial state transition graph and identified which nodes changed states and which did not, and compared them with the putative attractors. In all cases we find that the results are equivalent (or compatible for the five networks mentioned before). This result, together with our proof in Appendix \hyperref[AppendixProof]{A}, shows that every attractor of the system has a corresponding quasi-attractor of the reduction method.

Finally we compare the time performance of the methods. In Figure \ref{fig:RNtimes}(a) we show the average time to completion of the three methods on the same ensemble for different network sizes. Although at very small network sizes ($N\leq10$) the exact method (the whole state transition graph) is on average faster than the reduction method, for larger networks the reduction method is faster than the others. For large networks the reduction method is not only faster on average than the sampling method, but the distribution of times shows that, for all sizes, network reduction takes less than a second for most of the networks (see Figure \ref{fig:RNtimes}(b) for the $N=100$ case). This is true even at larger network sizes. For example, for $N=$ 100, 150, and 200 we have 88\%, 76\%, and 71\% of the networks, respectively, take less than one second. In contrast, for the sampling method none of the networks of these same sizes take less than one second. These results suggest that our method is not only more effective than state space sampling in the sense that it does not miss any of the attractors in the system, but it is also significantly more efficient in terms of time performance.

\begin{figure*}
\includegraphics{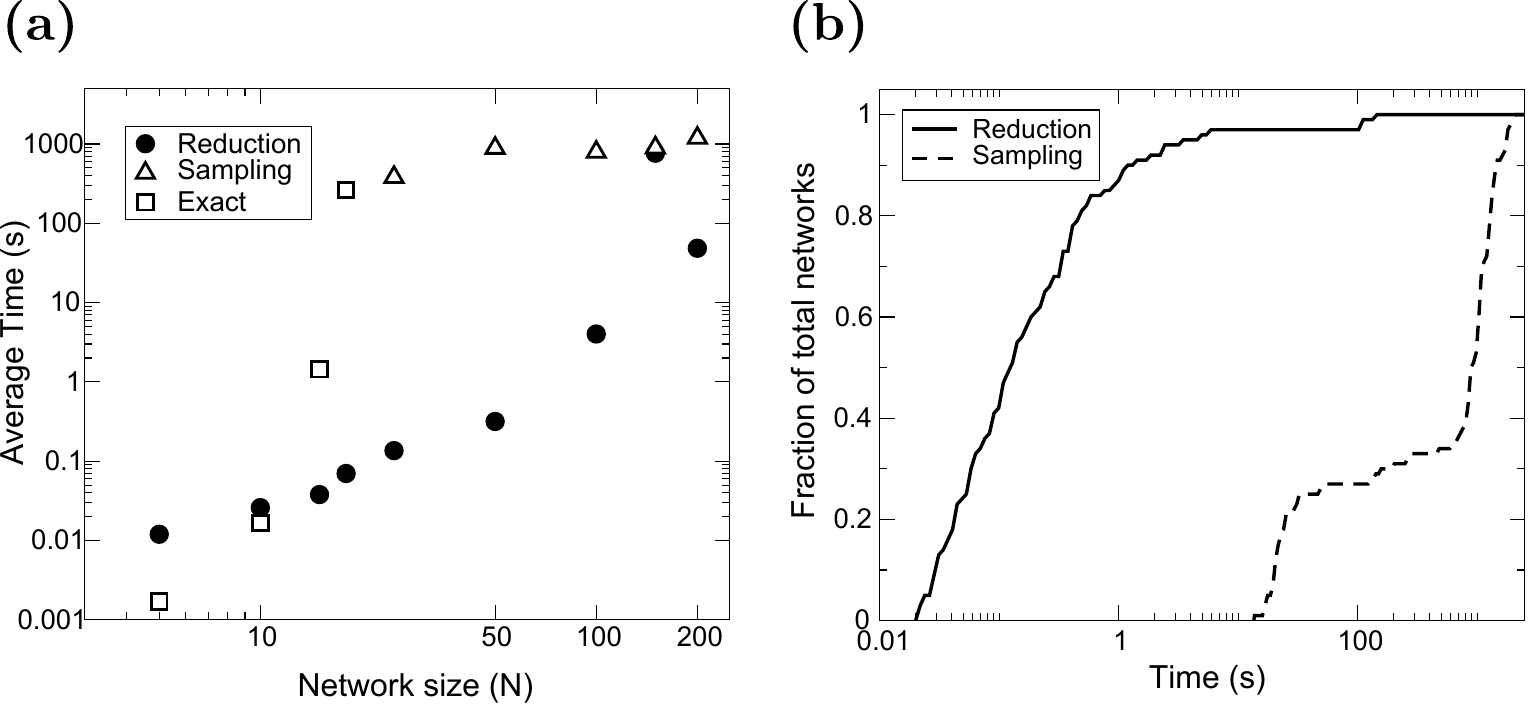}
\caption{Time performance of the different methods (see also the main text). (a) The average time it takes to find the attractors of a network for each method. Both axes are shown in a logarithmic scale. The bump shown in the $N=150$ case for the reduction method is the consequence of a network in the ensemble that took an unusually long time because of the large number of cycles in the network. (b) Cumulative distribution functions for the completion times in the $N=100$ ensemble. Note that the horizontal axis has a logarithmic scale.}
\label{fig:RNtimes}
\end{figure*}

\section{Discussion} \label{sec:Discussion}

\begin{figure}
\includegraphics{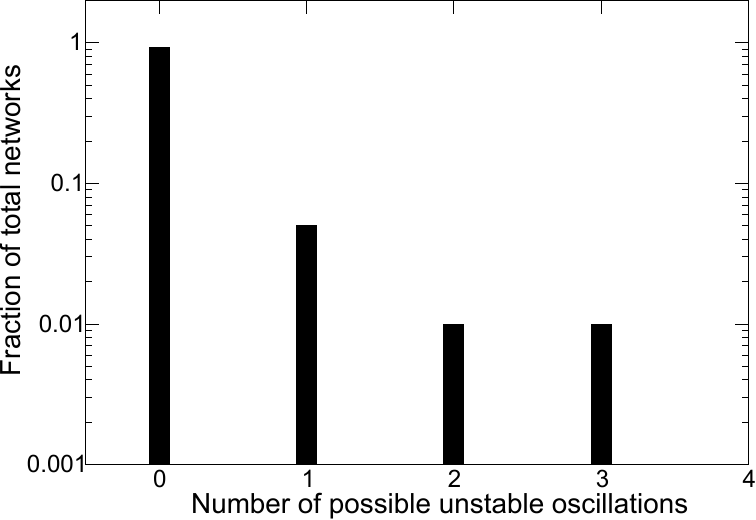}
\caption{Distribution function for the number of components that can display unstable oscillations in the $N=100$ ensemble. Note the logarithmic scale in the vertical axis. For approximately 90\% of the networks we find no such components. For the rest there are usually very few of them, with attractor sampling methods suggesting that none of them actually display unstable oscillations.}
\label{fig:UnstableAttrs}
\end{figure}

\begin{figure*}[t!]
\includegraphics{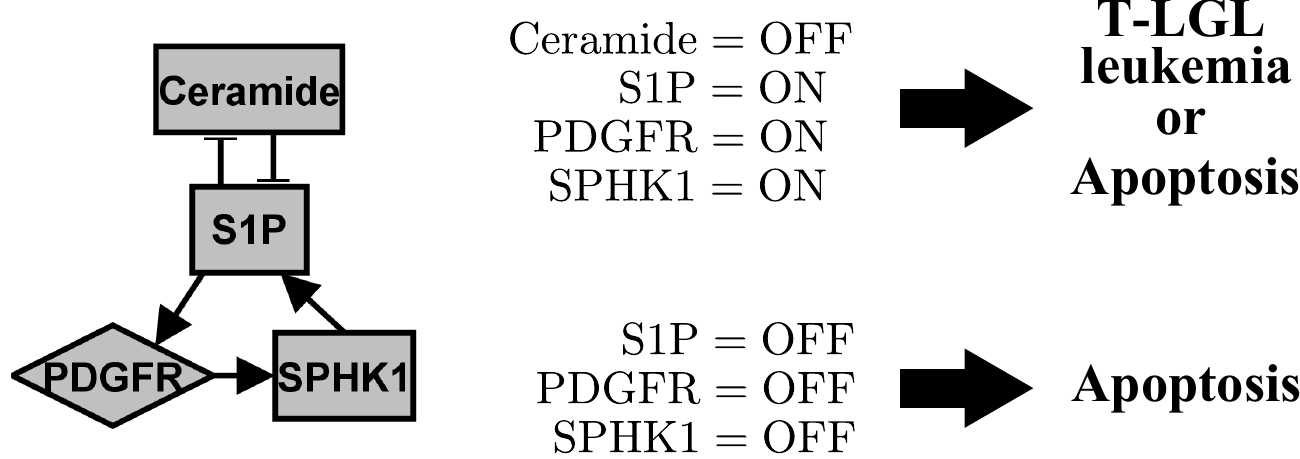}
\caption{The PDGFR-S1P-SPHK1-Ceramide motif, its allowed stable states, and the cell fates associated to them. For both set of stable states the apoptosis cell fate can be reached depending on the signals present, the asynchronous update order, and on the initial state. On the other hand, the T-LGL leukemia cell fate can only be reached if the motif stabilizes in the $\left\{\right.$PDGFR=S1P=SPHK1=ON, Ceramide=OFF$\left.\right\}$ state, regardless of the signals present, the asynchronous update order or of the initial state.}
\label{fig:PDGFRmotif}
\end{figure*}

In this work we have presented a novel reduction method that can greatly simplify a network, allowing us to deal with system sizes of an order of magnitude larger than what is possible through full state space searching methods. This reduction method, framed in the Boolean logic framework, uses an expanded representation of the network that explicitly includes the nature and logic of the interactions, which allows us to identify the network motifs that can stabilize in a steady state and use them to simplify the network. To test the validity of our method we applied it to a genetic network (the T-LGL leukemia survival network) and an ensemble of random Boolean networks of various sizes. We find that the results of our method always agree with the behavior of all networks tested.

An important point to make is that our reduction method does not actually reduce the complexity of the attractor-finding problem. What it does is to transfer the complexity of the problem from the state space size to how complex the network is, specifically, to how many cycles the network contains. The reasoning behind this transfer is that we want to take advantage of the sparseness of biological networks to make the problem more tractable for larger network sizes than a brute force approach (i.e. sampling the whole state space) would accomplish. Indeed, our results show that our reduction method is able to outperform a sampling of the state space both in terms of attractors/quasi-attractors found and in terms of computational time (see Figures \ref{fig:RNattractors} and \ref{fig:RNtimes}). Added to this is the fact that in not a single case did state space sampling find attractors not compatible with the results of network reduction, as was expected from our proof that the reduction method preserves all attractors.

A surprising result of applying our method to the ensemble of random networks was that for most cases it actually predicted the asynchronous dynamic repertoire of these networks, that is, it found which nodes stabilize and which of them oscillate in the attractors. This agreement is not trivial because our reduction method does not predict the state of the nodes remaining in the final reduced networks. What is the reason for this success? To answer this question, we need to remember that our reduction method is based on finding stable motifs and using them to simplify the network. Therefore, in the final reduced networks, all nodes that have adopted a fixed state are either part of a stable motif, or were stabilized by the influence of one of these motifs or by an input signal. If these were the only ways the state of a node can stabilize, then the remaining nodes would have to oscillate in the attractors. It is possible, however, (though not very likely) that a node adopts a fixed state without the influence of a stable motif. Specifically, this can happen when a node is downstream of an oscillating SCC that does not visit all possible states of their sub-state-space in an attractor. We discussed this kind of oscillatory behavior in section \ref{sec:2.7} and called it an incomplete oscillation.

The simplest way to obtain this kind of incompletely oscillatory attractors is to have two nodes with self loops that also form a feedback loop between themselves. For Kauffman $K=2$ networks one can show that this node configuration with the choice of Boolean rules that show these incomplete oscillations appears in a network with probability $1/32N(N-1)$. If one adds to this that the downstream nodes of these node configurations need to have very specific update rules in order to stabilize as a consequence of these oscillations, it is then not surprising that we found only very few of these cases in our network ensembles. More work in this direction is needed to generalize the reduction method to also identify these sets of fixed-state nodes.

Based on our discussion in section \ref{sec:2.7}, one may wonder if unstable oscillations are also a problem for our reduction method, and exactly how much of a problem they are. We have found that, for the ensemble of Kauffman critical networks, unstable oscillations are even more rare than incomplete oscillations. For example, the simplest and most probable way to have unstable oscillations (the one shown in the example in Figure \ref{fig:Unstableoscillations}) appears only with probability $1/128N(N-1)$ in a $K=2$ Kauffman network. To make sure that unstable oscillations are rare, we searched for the SCCs that could display this behavior.  For most networks ($\approx90\%$) we have found no such SCCs, and even when they exist, there are only a few of them, as illustrated on Figure \ref{fig:UnstableAttrs}. Not only that, but the few that have the possibility to display unstable oscillations do not actually seem to do so (satisfying property (1) and (2) but not property (3) of the conditions stated in section \ref{sec:2.7} is only a sufficient condition to display unstable oscillations), as none of them were found by the attractor sampling methods in these networks. It is noteworthy that these SCCs that have the possibility to display unstable oscillations are responsible for the marked quasi-attractors referred to in the network reduction algorithm in Appendix \hyperref[AppendixAlgorithm]{B}. More work in this direction is needed to find more stringent conditions that identify the unstable oscillations.

Using our network reduction method on the T-LGL leukemia network we were able to find the recurring stable motifs (Figure \ref{fig:TLGLMotifs}). A natural question to ask is if the stable motifs found during network reduction have any special biological significance. Indeed, all of these motifs do play a significant role in the biology of T-LGL leukemia: PDGFR-S1P-SPHK1-Ceramide (ceramide/sphingomyelin pathway) has been shown to be essential for T-LGL cell survival \cite{PDGFRTLGL}; moreover, TBET (T-box transcription factor) and IFNG-P2 are related to the control of two of the main cytokines produced by cytotoxic T cells (interleukin 2 and interferon gamma, respectively), whose low production is one of the characteristics of T-LGL leukemia \cite{TLGLPNAS,IL2TLGL,IFNGTLGL}.

An interesting observation is that these three motifs are directly regulated by five of the input signals of the T-LGL network (or almost directly in the case of TBET, see Figure \ref{fig:TLGLMotifs}), which suggest their importance in cell-fate determination for cytotoxic T cells. This appears to be especially true for the PDGFR-S1P-SPHK1-Ceramide motif, whose components need to always be in a specific set of states (PDGFR=S1P=SPHK1=ON, Ceramide=OFF) for the T-LGL leukemia cell fate to be possible (see Figure \ref{fig:PDGFRmotif}). This is consistent with and actually explains the previous finding by Zhang et al. that an intermittent signal of PDGF (coupled with the sustained presence of IL15) was enough for T-LGL leukemia to be possible: this happens because the intermittent signal is enough to stabilize the PDGFR-S1P-SPHK1-Ceramide motif in the required state for the T-LGL leukemia cell fate to become possible.

To summarize, our study showed that the novel network reduction method we propose allows us to overcome the limitations related to the vast state space of large networks by taking advantage of the stable components naturally present in biological networks. This is accomplished by transferring the complexity of the problem from the size of the state space to the complexity of the network, namely, the number of cycles it has. For most cases, we also found that our method goes beyond reducing the network size and can actually predict the asynchronous dynamical repertoire in the attractors of the system. For the case of the T-LGL leukemia network we found that the stable components identified by our method play an important role in the biology of T-LGL leukemia and appear to be used as a cell-fate determination mechanism for cytotoxic T cells. Overall, our method adds a powerful technique to the set of tools available to infer the dynamical behavior of a network based on the topology and the nature of the interactions, a technique that is flexible enough that it can be applied to a large variety of biological networks.

\section*{Acknowledgements} \label{Acknowledgements}

This work was supported by NSF grants IIS 1161001 and PHY 1205840. We would also like to thank Zhongyao Sun and Rui-Sheng Wang for fruitful discussions.

\section*{Appendix} \label{Appendix}

\subsection*{Appendix A - Proof of the conservation of attractors by the expanded network/network reduction method} \label{AppendixProof}

In the following we use $V=\left( v_1, v_2, \ldots, v_N \right)$ to represent the $N$ nodes of the Boolean network, $\Sigma=\left(\sigma_1, \sigma_2, \ldots, \sigma_N \right)$ to represent the states of these nodes, and $F=\left(f_1, f_2 , \ldots, f_N \right)$ to represent the Boolean functions associated to each of these nodes.

We assume, for convenience, that the Boolean functions in $F$ satisfy these three properties:
\begin{enumerate}
  \item The $f_i$'s do not take constant values (i.e. $f_i\neq0$ and $f_i\neq1$).
  \item If $f_i$ depends on the state of node $v_j$, $\sigma_j$, then there must be at least a pair of states $\Sigma^{(1)}$ and $\Sigma^{(2)}$ with $\sigma_j^{(1)} \neq \sigma_j^{(2)}$, and $\sigma_k^{(1)} = \sigma_k^{(2)}$ for $k \neq j$, such that $f_i(\Sigma^{(1)}) \neq f_i(\Sigma^{(2)})$. This is equivalent to requiring that the Boolean derivative of $f_i$ with respect to $\sigma_j$ is nonzero for at least a pair of network states \cite{BooleanDeriv}.
  \item If, for a state of a subset of the inputs of $f_i$, one has $f_i=1$ (whatever the states of the remaining inputs), then the disjunctive form of $f_i$ must have at least one of its conjunctive clauses equal to 1 when evaluated at the state of this subset of nodes.
\end{enumerate}
The first property makes sure we have no source nodes. For our purposes this can be assumed without loss of generality, because even if that is not the case, we can use the reduction method of Saadatpour et al. \cite{AssiehMath} and remove all source nodes while preserving all attractors. The second property can also be assumed without any loss of generality; it is just a way of stating that we consider $f_i$ to depend on $\sigma_j$ only if it explicitly depends on $\sigma_j$ for at least a pair of network states. The third property is also general, since one can construct the respective disjunctive normal form from the truth table of the Boolean function.

Our first proposition states that the stable motifs found from the expanded network are such that the corresponding states of these motifs are partial fixed points of the Boolean rules of the nodes involved.
\begin{myprop} \label{PartFixedPoints}
Let $M=\left(V_{m_1}, V_{m_2}, \ldots, V_{m_l},\right.$$\left.V_{m_{l+1}}, V_{m_{l+2}}, \ldots, V_{m_L}\right)$ be a stable motif in the expanded network representation, where $V_{m_1}, V_{m_2}, \ldots, V_{m_l}$ can either be a normal node or a complementary node, and where $V_{m_{l+1}}, V_{m_{l+2}}, \ldots, V_{m_L}$ are composite nodes. We denote $M_{state}=\left(\sigma_{m_1}=b_{m_1},\right.$$\left.\sigma_{m_2}=b_{m_2}, \ldots, \sigma_{m_l}=b_{m_l} \right)$, with $b_{m_j} \in \left\{0, 1\right\}$ as the corresponding state of $M$ in the network state $\Sigma$: $b_{m_j}=1$ if it is a normal node, and $b_{m_j}=0$ if it is a complementary node. Then, for any normal node $v_{m_j}$ or complementary node $\overline{v}_{m_j}$ in $M$ and for any network state $\Sigma_M$ in which $\sigma_{m_k}=b_{m_k} \forall m_k \in \left\{m_1, m_2, \ldots, m_l\right\}$, we will have $f_{m_j}(\Sigma_M)=b_{m_j}$.
\end{myprop}

Let us sketch the proof for this proposition. First, because a stable motif only contains a node or its complementary node, one can do a change of variables in the original network so that the state of the nodes in $M$ is $1$ in $M_{state}$. This simplifies the problem, since we now just need to show that $f_{m_j}(\Sigma_M)=1$. Now, the Boolean function of node $v_{m_j}$ has the form $f_{m_j}=S_1\hbox{ OR }S_2\hbox{ OR }\cdots\hbox{ OR }S_n$, where $S_i=s_1\hbox{ AND }s_2\hbox{ AND }\cdots\hbox{ AND }s_I$ and where the $s_k$'s are either a node state or its negation. Since every node $v_{m_j}$ in $M$ has at least an input from another node in $M$, then this means that one of the $S_i$'s of $f_{m_j}$ corresponds to this node input. If we call $S_j$ the corresponding $S_i$, then it must be such that all the $s_k$'s of this $S_j$ must be the states of nodes in $M$. As a consequence $S_{j}(\Sigma_M)=1$, since in the state $\Sigma_M$ all the states of nodes in $M$ is 1.

The reverse of this proposition is also true, that is, if for a given set of node states updating any of the states in the set gives back the same state, regardless of the state of any node outside of the set, then this set of states will correspond to a set of stable motifs in the expanded network representation:

\begin{myprop} \label{Stablemotifs-fixedpoints}
Let $M_{state}=\left(\sigma_{m_1}=b_{m_1}, \sigma_{m_2}=b_{m_2},\right.$ $\left.\ldots, \sigma_{m_l}=b_{m_l} \right)$ be the state of a set of nodes such that if $\Sigma_M$ is any network state in which $\sigma_{m_k}=b_{m_k} \forall m_k \in \left\{m_1, m_2, \ldots, m_l\right\}$, then $f_{m_j}(\Sigma_M)=b_{m_j}$. Then (i) there is a set of stable motifs $\left\{M_n\right\}$ in the expanded network representation such that each of the $M_n$'s contain only normal nodes or complementary nodes of the nodes whose state is specified in $M_{state}$ (normal nodes if $b_{m_k}=1$, and complementary nodes if $b_{m_k}=0$) and in which all other nodes in the $M_n$'s (if any) will be composite nodes made up of the normal nodes or complementary nodes in the corresponding $M_n$, and (ii) the nodes whose state is specified in $M_{state}$ but that are not included in the set of stable motifs $\left\{M_n\right\}$ will be downstream of the nodes in at least one of the stable motifs.
\end{myprop}

Part of the proof for this proposition is very similar to the one of Proposition \ref{PartFixedPoints}. First, one does the same change of variables and writes down the Boolean function of an arbitrary element of the nodes whose state is specified in $M_{state}$. Then, from the form of the Boolean function and since $f_{m_j}(\Sigma_M)=1$, at least one of the conjunctive clauses of this Boolean function will be composed only of a normal node of the nodes whose state is specified in $M_{state}$, or composite nodes composed of these normal nodes. Finally, if one creates the network composed only of these normal nodes and composite nodes, and separates them into SCCs, one will find a set of source SCCs. Since a source SCC contains all of its inputs (by definition), and the elements of these SCCs contain only normal nodes and composite nodes composed of these normal nodes (by construction), then these source SCCs will be the stable motifs we are looking for.

For the next propositions we need to remember certain properties of the attractors of the asynchronous updating scheme. For any attractor $\mathcal{A}$, we can divide the $N$ nodes into two classes: those that take the same value in all network states of $\mathcal{A}$ (i.e, either 0 or 1), and those that take more than one value in the different network states of $\mathcal{A}$ (i.e, both 0 and 1). We refer to the former as \emph{stabilized nodes}, and to the latter as \emph{oscillating nodes}. The following propositions state that stabilized nodes can have inputs from stabilized nodes or oscillating nodes (Proposition \ref{Stablemotifs-inputs}), while oscillating nodes must have at least one oscillating node as an input (Proposition \ref{Oscmotifs-inputs}).

\begin{myprop} \label{Stablemotifs-inputs}
Let $\mathcal{A}$ be an attractor of the Boolean network $(V,\Sigma,F)$, and let $\mathcal{S}$ and $\mathcal{O}$ be the set of the stabilized and oscillating nodes in the attractor, respectively. If $v_{s} \in \mathcal{S}$, and $b_s$ is the stabilized state of node $v_{s}$, then one of the following two cases holds: (i) one of the conjunctive clauses of $f_s$ (if $b_s=1$) or $\overline{f}_s$ (if $b_s=0$) depends only on the specific state of the nodes of $\mathcal{S}$ in $\mathcal{A}$. If (i) is not true, then (ii) for both $f_s$ and $\overline{f}_s$ at least one conjunctive clause depends on the state of one or more nodes in $\mathcal{O}$ and, if the clause depends on any more states, they have to be the state of the nodes of $\mathcal{S}$ in $\mathcal{A}$.
\end{myprop}

\begin{myprop} \label{Oscmotifs-inputs}
Let $\mathcal{A}$ be an attractor of the Boolean network $(V,\Sigma,F)$, and let $\mathcal{S}$ and $\mathcal{O}$ be the set of the stabilized and oscillating nodes, respectively. If $v_{o} \in \mathcal{O}$ then (i) neither $f_o$ nor $\overline{f}_o$ can have any conjunctive clauses that depend only on the state of the nodes of $\mathcal{S}$ in $\mathcal{A}$ (i.e, on $\sigma_s$ if $b_s=1$, or $\overline{\sigma}_s$ if $b_s=0$), and (ii) both $f_o$ and $\overline{f}_o$ must have at least one conjunctive clause that depends on the state of one or more nodes in $\mathcal{O}$ and, if this same clause depends on any other states, they must be the states of nodes of $\mathcal{S}$ in $\mathcal{A}$.
\end{myprop}

To illustrate the implications of Propositions \ref{Stablemotifs-inputs} and \ref{Oscmotifs-inputs}, consider the three node network with nodes $A$, $B$ and $C$, and the following Boolean functions (and their Boolean negations)
\begin{align*}
    f_A =& \hbox{NOT }A\hbox{ OR }\hbox{NOT }B\hbox{ OR }C, \\
    f_B =& \hbox{NOT }A\hbox{ OR }\hbox{NOT }B\hbox{ OR }C, \\
    f_C =& (A\hbox{ AND }B)\hbox{ OR }C, \\
    \overline{f}_A =& A\hbox{ AND }B\hbox{ AND NOT }C, \\
    \overline{f}_B =& A\hbox{ AND }B\hbox{ AND NOT }C, \\
    \overline{f}_C =& (\hbox{NOT }A\hbox{ AND }\hbox{NOT }C)\hbox{ OR }(\hbox{NOT }B\hbox{ AND }\hbox{NOT }C).
\end{align*}
Note that these Boolean functions satisfy the three properties for a Boolean function in $F$ outlined at the beginning of this Appendix. For this network there is an attractor $\mathcal{A}$ with $A$ and $B$ oscillating and $C=0$. For this attractor the set of stabilized and oscillating nodes is $\mathcal{S}=\{C\}$ and $\mathcal{O}=\{A, B\}$, respectively. Point (i) of Proposition \ref{Stablemotifs-inputs} states that, for the attractor $\mathcal{A}$, one possibility is that a conjunctive clause of $\overline{f}_C$ has only the the specific state of nodes of $\mathcal{S}$ in $\mathcal{A}$ (i.e., $\hbox{NOT }C$, since $C=0$ in $\mathcal{A}$). Since none of the conjunctive clauses of $\overline{f}_C$ satisfies point (i), Proposition \ref{Stablemotifs-inputs} states that $f_C$ and $\overline{f}_C$ must satisfy point (ii). For the case of $f_C$ the clause $A\hbox{ AND }B$ depends on at least a node in $\mathcal{O}$, so it does satisfy point (ii). For the case of $\overline{f}_C$, any of the two clauses ($\hbox{NOT }A\hbox{ AND }\hbox{NOT }C$ or $\hbox{NOT }B\hbox{ AND }\hbox{NOT }C$) are enough to satisfy point (ii) since both clauses depend on at least one node in $\mathcal{O}$ ($A$ and $B$, respectively), and the other states they depend on is the state of $C$ in $\mathcal{A}$.

For the network and attractor used in the previous paragraph, Proposition \ref{Oscmotifs-inputs} states that $f_A$, $f_B$, $\overline{f}_A$, and $\overline{f}_B$ must satisfy two properties. In this network we have $f_A=f_B$, so we only need to consider $f_A$ and $\overline{f}_A$. Property (i) requires that neither $f_A$ nor $\overline{f}_A$ can have a conjunctive clause that contains only the term $\hbox{NOT }C$, which is indeed the case. Property (ii) requires that both $f_A$ and $\overline{f}_A$ must have at least one conjunctive clause with a state of one or more nodes in $\mathcal{O}$, which is the case since $f_A$ has the clauses $\hbox{NOT }B$ and $\hbox{NOT }A$, and $\overline{f}_A$ has the clause $A\hbox{ AND }B\hbox{ AND NOT }C$. For the clause $A\hbox{ AND }B\hbox{ AND NOT }C$, we have that it depends on the state of a node not in $\mathcal{O}$ ($\hbox{NOT }C$), so property (ii) also requires that the state it depends on must be the state of a node of $\mathcal{S}$ in $\mathcal{A}$ (i.e., $\hbox{NOT }C$, since $C=0$ in $\mathcal{A}$), which is the case.

We now proceed to prove the three lemmas that will allow us to show that the reduction method conserves all attractors. In Lemma \ref{Stablemotifs-reduction1} we construct the set of nodes, for an arbitrary attractor, whose state will be identified by our reduction method, $\mathcal{S}_{red} \subset \mathcal{S}$. We also show that there is at least one stable motif composed of the corresponding states in the attractor of the nodes of $\mathcal{S}_{red}$ (as long as $\mathcal{S}_{red}$ is not empty). In Lemma \ref{Stablemotifs-reduction2} we show that the network reduction of these stable motifs can only stabilize nodes in $\mathcal{S}_{red}$.

\begin{mylemma} \label{Stablemotifs-reduction1}
Let $\mathcal{A}$ be an attractor of the Boolean network $(V,\Sigma,F)$, and let $\mathcal{S}$ and $\mathcal{O}$ be the set of the stabilized and oscillating nodes of $\mathcal{A}$, respectively. There exists a set of nodes $\mathcal{S}_{red} \subset \mathcal{S}$ such that in the expanded network representation of $(V,\Sigma,F)$ there will be at least one stable motif composed only of the corresponding states of the nodes of $\mathcal{S}_{red}$ in $\mathcal{A}$, or composite nodes composed of such nodes.
\end{mylemma}
\begin{proof}
Without loss of generality we can do a change of variables so that $\sigma_{s}=1$ if $v_s \in \mathcal{S}$. By Proposition \ref{Stablemotifs-inputs}, we can divide $\mathcal{S}$ into the nodes that have a conjunctive clause in their rule that depends only on the specific state of nodes of $\mathcal{S}$ in $\mathcal{A}$, or the nodes that have at least a conjunctive clause in their rule that depends on the state of a node in $\mathcal{O}$. We will refer to the former as $\mathcal{S}_{0}$ and to the latter as $\mathcal{S}_{osc}$.

Let $\mathcal{S}_{1} \subset \mathcal{S}_{0}$  be the nodes that have at least one conjunctive clause in their rules that depends only on the specific state of the nodes of $\mathcal{S}_{0}$ in $\mathcal{A}$ (i.e, on $\sigma_s$, because of the change of variables). Let $\mathcal{S}_{2} \subset \mathcal{S}_{1}$  be the nodes that have at least one conjunctive clause in their rules that depends only on the specific state of nodes of $\mathcal{S}_{1}$ in $\mathcal{A}$ (note they could depend on the states of nodes in $\mathcal{S}_{0} - \mathcal{S}_{1}$). We do this iteratively until $\mathcal{S}_{i_{max}+1}=\mathcal{S}_{i_{max}}$ and denote $\mathcal{S}_{red}=\mathcal{S}_{i_{max}} \subset \mathcal{S}_{0}$. Since $\mathcal{S}_{red}$ was constructed by first removing the nodes that required nodes in $\mathcal{S}_{osc}$ to stabilize, and then removing the ones that depended on the previously reduced nodes, and so on, then $\mathcal{S}_{red}$ corresponds to the set of nodes in $\mathcal{S}_{0}$ that do not depend in any way on nodes of $\mathcal{S}_{osc}$ to stabilize in their state on $\mathcal{A}$.

We can now show that the expanded network representation of the states of the nodes of $\mathcal{S}_{red}$ in $\mathcal{A}$ has at least one stable motif composed only of the corresponding states of $\mathcal{S}$ in $\mathcal{A}$ or composite nodes composed of such nodes. Note that because of our change of variables, we only need to consider normal nodes in $\mathcal{S}$ and not complementary nodes of the nodes in $\mathcal{S}$. We first note that the construction of $\mathcal{S}_{red}$ makes sure that its nodes have at least one conjunctive clause in their rules that depends only on the state of nodes of $\mathcal{S}_{red}$ in $\mathcal{A}$. As a consequence, and since the stabilized state was taken to be 1, the expanded network representation of the state of the nodes of $\mathcal{S}_{red}$ in $\mathcal{A}$ will have an input from either the nodes of the corresponding states of $\mathcal{S}_{red}$ in $\mathcal{A}$, and/or composite nodes composed only of said nodes.

From the expanded network representation we take the nodes of the corresponding states of $\mathcal{S}_{red}$ in $\mathcal{A}$ (normal nodes) and the composite nodes composed only of said nodes, and construct a network $\mathcal{V}$ with these normal nodes, composite nodes, and the edges between them. From the discussion in the previous paragraph, each of the nodes in $\mathcal{V}$ must have, at least, one input node. This means that if we divide $\mathcal{V}$ into SCCs there will be at least one SCC in which the nodes have no inputs outside the SCC itself (a source SCC).  These source SCCs are stable motifs since the composite nodes in each of these SCCs have all their inputs included (the nodes in these SCCs have no inputs outside the SCC itself) and by construction these SCCs contain no complementary nodes.
\end{proof}

\begin{mylemma} \label{Stablemotifs-reduction2}
Let $\mathcal{S}_{red} \subset \mathcal{S}$ be the constructed set of nodes in Lemma \ref{Stablemotifs-reduction1}. Then (i) $\mathcal{S}_{red}$ is such that the network reduction of any stable motif composed only of the corresponding states of $\mathcal{S}_{red}$ in $\mathcal{A}$ (or composite nodes composed of such nodes) can only stabilize nodes in $\mathcal{S}_{red}$, and (ii) if any of the states of the nodes in $\mathcal{S}_{red}$ stabilizes, then it has to be on their corresponding state in $\mathcal{A}$; if they do not stabilize, then either their rule (if their stabilized state in $\mathcal{A}$ is 1) or the negation of their rule (if their stabilized state is 0) will have a conjunctive clause that only depends on the specific state of the nodes of $\mathcal{S}_{red}$ in $\mathcal{A}$ (i.e, on $\sigma_s$ if $b_s=1$, or $\overline{\sigma}_s$ if $b_s=0$) that did not stabilize during network reduction.
\end{mylemma}
\begin{proof}
Without loss of generality we will again do a change of variables so that $\sigma_{s}=1$ if $v_s \in \mathcal{S}$. Let us start with the nodes in $\mathcal{S}_{0} - \mathcal{S}_{red}$ (the stabilized nodes not in $\mathcal{S}_{red}$), as constructed in Lemma \ref{Stablemotifs-reduction1}. These nodes cannot have a conjunctive clause in their rules that depend only on the specific state of the nodes of $\mathcal{S}_{red}$ in $\mathcal{A}$, since that would make them part of $\mathcal{S}_{red}$. This means that setting the states of the nodes of $\mathcal{S}_{red}$ in $\mathcal{A}$ in their rules will not set the value of the rule to 1. It can also not make it be equal to 0, since each of them has a conjunctive clause in their rule that can be equal to 1 if the state of the rest nodes they have as an input (which have not been set to any value) take any of the states in $\mathcal{A}$. Hence, the nodes $\mathcal{S}_{0} - \mathcal{S}_{red}$ will not be stabilized by the reduction of any of the stable motifs being considered.

Let us now look at what happens to the nodes in $\mathcal{S}_{osc}$ during the reduction of the stable motifs of interest. In Proposition \ref{Oscmotifs-inputs} we showed that for every node $v_o \in \mathcal{S}_{osc}$ neither $f_o$ nor $\overline{f}_o$ could have any of their conjunctive clauses depend only on the specific states of the nodes of $\mathcal{S}$ in $\mathcal{A}$. As a consequence, none of the nodes in $\mathcal{S}_{osc}$ will be stabilized by the reduction of stable motifs  made up only of nodes in $\mathcal{S}_{red} \subset \mathcal{S}$. Since we have now shown than neither the nodes in $\mathcal{S}_{0} - \mathcal{S}_{red}$ nor the nodes in $\mathcal{S}_{osc}$ can be stabilized by the reduction of the nodes in $\mathcal{S}_{red}$ in their state in $\mathcal{A}$, then point (i) of the lemma has been proved.

To show point (ii), let us consider the iterative process involved in network reduction. First, one sets the states of the nodes in the chosen source SCC in the rules of all nodes in $\mathcal{S}_{red}$. As a consequence, the nodes in $\mathcal{S}_{red}$ that have any conjunctive clause in their rule depending only on the specific state of the nodes in the source SCC will stabilize in the 1 state. The rest of the nodes in $\mathcal{S}_{red}$ cannot stabilize on either 1 (or they would be part of the previous nodes) nor on 0 (since they still have an conjunctive clause in their rule that depends on specific the state of the nodes of $\mathcal{S}_{red}$ in $\mathcal{A}$ not in the source SCC). If one now evaluates the states of the nodes that just stabilized on 1 in the rest of the nodes in $\mathcal{S}_{red}$, and follows the same arguments, the result will be a new set of nodes that just stabilized on 1, and a set of nodes that have, at least, one conjunctive clause in their rule only depending only on the specific state of the nodes of $\mathcal{S}_{red}$ in $\mathcal{A}$ that have not stabilized. Doing this iteratively until no more nodes stabilize one finds the desired result, that is, that the nodes in $\mathcal{S}_{red}$ either stabilize at their state in $\mathcal{A}$ or if they do not, then they still have, at least, one conjunctive clause that now depends only on the specific state of the nodes of $\mathcal{S}_{red}$ in $\mathcal{A}$ that this reduction did not stabilize.
\end{proof}

\begin{mylemma} \label{Oscillatingmotifs-reduction}
Let $\mathcal{A}$ be an attractor of the Boolean network $(V,\Sigma,F)$, and let $\mathcal{S}$ and $\mathcal{O}$ be the set of the stabilized and oscillating nodes, respectively. Let $\mathcal{S}_{red} \subset \mathcal{S}$ be the constructed set of nodes in Lemma \ref{Stablemotifs-reduction1} and assume that $\mathcal{S}_{red}$ is empty and that $\mathcal{O}$ is a non empty set. Then the expanded network representation of $(V,\Sigma,F)$ must be such that the normal nodes and complementary nodes of the elements in $\mathcal{O}$, and the nodes corresponding to the state of the nodes of $\mathcal{S}$ in $\mathcal{A}$ must both be downstream of an oscillating motif that contains at least one of the nodes in $\mathcal{O}$.
\end{mylemma}
\begin{proof}
We first show that the expanded network has, at least, one source SCC which has either a normal node or complementary node of a node in $\mathcal{O}$. First, since the Boolean network has no source nodes, then the expanded network representation will also have no source nodes. We can then divide the expanded network into SCCs, with at least one of these being a source SCC. Because of the absence of source nodes, all nodes in the expanded network will be downstream of one of these source SCCs. Let us assume that one of these source SCCs has neither a normal node nor a complementary node of the elements in $\mathcal{O}$. This means that all the nodes in this source SCC are either normal nodes or complementary nodes of the elements in $\mathcal{S}$. But, since we assumed $\mathcal{S}_{red}$ is empty, the construction of $\mathcal{S}_{red}$ in Lemma \ref{Stablemotifs-reduction1} guarantees that all nodes in $\mathcal{S}$ are downstream of a normal node or complementary node of the elements in $\mathcal{O}$. Therefore, this source SCC must contain a normal or complementary node of the elements in $\mathcal{O}$, which is a contradiction. Therefore, all source SCCs contain either a normal node or a complementary node of the elements in $\mathcal{O}$. As a consequence, all nodes in the expanded network are downstream of these source SCCs.

Before proceeding, we first need to note a certain property of the Boolean function with the properties we specified at the beginning of the section. Without loss of generality, $f$ can be written in the form $f=S_1\hbox{ OR }S_2\hbox{ OR }\cdots\hbox{ OR }S_n$, where $S_i=s_1\hbox{ AND }s_2\hbox{ AND }\cdots\hbox{ AND }s_I$, and where each $s_j$ is either a node state or the negation of a node state. Now, lets assume that $s_k$ is one of the $s_j$'s of $f$, then $\overline{f}$ will have one of its corresponding $s_j$'s be the negation of $s_k$. This can be proved by using property 3 of the Boolean functions.

Let $V_O$ be any of the source SCCs that contains a normal or complementary node of the elements in $\mathcal{O}$. Because of the property discussed in the previous paragraph, the complement of the normal nodes or complementary nodes of $V_O$ will also form a source SCC. These two SCCs have to be connected to each other, or these nodes would not be able to oscillate in the attractor $\mathcal{A}$. Since they are source SCCs, this means the are actually part of the same source SCC. This shows that $V_O$ contains both the nodes and complementary nodes of every element it contains. Finally, since these source SCCs were constructed by separating the whole expanded network into SCCs, then they are the largest SCC (i.e. the usual definition of SCC). Hence, we have shown that these source SCCs are oscillating motifs.
\end{proof}

The following theorem is the main result of this section, and it combines the results of Lemma \ref{Stablemotifs-reduction1}, \ref{Stablemotifs-reduction2}, and \ref{Oscillatingmotifs-reduction}. It shows that for every attractor in the network our reduction method will find a corresponding quasi-attractor in which the state of the nodes in $\mathcal{S}_{red}$ is the same as in the attractor, and in which the rest of the nodes will either be part of an oscillating motif or downstream of it.

\begin{mythm} \label{Attractorconservation}
Let $\mathcal{A}$ be an attractor of the Boolean network $(V,\Sigma,F)$, and let $\mathcal{S}$ and $\mathcal{O}$ be the set of the stabilized and oscillating nodes, respectively. Let $\mathcal{S}_{red} \subset \mathcal{S}$ be the set of nodes constructed in Lemma \ref{Stablemotifs-reduction1}. Then, there exists a set of stable motifs such that, by applying network reduction, all the nodes in $\mathcal{S}_{red}$ will stabilize in their steady state in $\mathcal{A}$, while the rest of the nodes in $V$ will be part of the final reduced network. This resulting final reduced network will be such that, in its expanded network representation, all the nodes will either be part of an oscillating motif containing at least one of the nodes in $\mathcal{O}$, or be downstream of an oscillating motif.
\end{mythm}
\begin{proof}
Using Lemma \ref{Stablemotifs-reduction2}, the network obtained after reducing any stable motif composed only of the corresponding states of $\mathcal{S}_{red}$ in $\mathcal{A}$ will have a new $\mathcal{S}_{red}$ containing only the nodes in the previous $\mathcal{S}_{red}$ that did not stabilize. If one performs network reduction using the stable motif that necessarily exists in the new network (because of Lemma \ref{Stablemotifs-reduction1}), and does this iteratively, one will obtain a network where $\mathcal{S}_{red}$ is empty and where only the states of the nodes in the original $\mathcal{S}_{red}$ stabilized during reduction in their state in $\mathcal{A}$. By the results in Lemma \ref{Oscillatingmotifs-reduction}, this resulting network has a set of oscillating motifs with at least one of the nodes in $\mathcal{O}$, and with the rest of the nodes downstream of these oscillating motifs.
\end{proof}

\subsection*{Appendix B - The full network reduction algorithm} \label{AppendixAlgorithm}

In the following we describe the full network reduction algorithm. During the description of the algorithm we refer the reader to the subsections in section \ref{sec:2} where each of these steps are described in more detail. A Java implementation of the network reduction algorithm is available per request to the interested reader.

\begin{enumerate}
\item For every combination of the states of the source nodes (nodes with no upstream components) apply the two steps of network reduction method described in section \ref{sec:2.6} recursively until neither of them can be applied anymore.
\item Create the expanded network representation for each of the resulting networks (section \ref{sec:2.4}).
\item Search the expanded network for stable motifs (section \ref{sec:2.5}) and oscillating components (section \ref{sec:2.7}).
\item For every separate stable motif create a copy of the current network. On each of the networks created use the states of the corresponding stable motif as inputs and apply the two steps of the network reduction described in section \ref{sec:2.6} recursively until neither of them can be applied anymore.
\item For every oscillating component of more than two nodes (i.e., every oscillating component that could display incomplete oscillations) create a copy of the current network. On each of the networks created, the nodes in the corresponding oscillating component and the nodes downstream of this component will be marked. The marked nodes cannot be reduced at any later step of the algorithm (i.e, they will have their state undetermined in the quasi-attractors that are derived from these networks).
\item For the oscillating components of two nodes (i.e, only one normal node and its corresponding complementary node), check if any node downstream of these oscillating motifs participates in a stable motif with no composite nodes. If any of them do, go to step 7; otherwise, check if there are any stable motifs that are downstream of these oscillating components (these stable motifs would necessarily have a composite node). If there are not, go to step 7; if there are, check if any of them is downstream of a stable motif that is itself not downstream of any of these oscillating components. If this is the case, go to step 7; if this is not the case, then create one copy of the current network and mark the nodes in the oscillating motifs considered in this step and the nodes downstream of them. The marked nodes cannot be reduced at any later step of the algorithm (i.e, they will have their state undetermined in the quasi-attractors that are derived from these networks).
\item Repeat 2, 3, 4, 5 and 6 for each of the networks iteratively until no more stable motifs are found. The result, a set of fixed state nodes and their stabilized states, and a set of nodes with undetermined states with their reduced Boolean functions, is the set of quasi-attractors (section \ref{sec:2.6}).
\item Prune the set of quasi-attractors of duplicates (two quasi-attractors are the same if they have the same set of fixed state nodes and the same state for these stabilized nodes; if two quasi-attractors are the same, except that one of them has some nodes marked while the other one does not, remove the one that has the marked nodes).
\end{enumerate}

Some of the resulting quasi-attractors will have marked nodes while others will not. For every unmarked quasi-attractor there will necessarily be a corresponding attractor in the Boolean network. For a marked quasi-attractor there may not be a corresponding attractor in the Boolean network; only by knowing the specific states visited during oscillations by the undetermined nodes in the quasi-attractor's reduced network can one confirm whether there is a corresponding attractor (this is a consequence of incomplete oscillations and unstable oscillations, see section \ref{sec:2.7}).

\subsection*{Appendix C - The logical rules of the T-LGL leukemia survival network} \label{AppendixC}

Logical rules governing the state of the T-LGL leukemia survival signaling network depicted in Figure \ref{fig:TLGLnetwork}. For simplicity, the nodes' states are represented by the node names. The Boolean rules were constructed based on experimental results of the corresponding cellular elements in healthy and leukemic cytotoxic T cells, in such a way that that the model reproduces the result of knockout and overexpression experiments. The 'NOT Apoptosis' clause in each rule implements the fact that in the apoptosis (cell death) state every node except Apoptosis is OFF. This table is adapted from \cite{AssiehPCB,TLGLPNAS}. The interested reader is referred to \cite{TLGLPNAS} for the detailed explanation of the rules.
\ \\
\ \\
$f_{CTLA4}$ = TCR AND NOT Apoptosis \\
$f_{TCR}$ = Stimuli AND NOT (CTLA4 OR Apoptosis) \\
$f_{PDGFR}$ = (S1P OR PDGF) AND NOT Apoptosis \\
$f_{FYN}$ = (TCR OR IL2RB) AND NOT Apoptosis \\
$f_{Cytoskeleton\hbox{ }signaling}$ = FYN AND NOT Apoptosis \\
$f_{LCK}$ = (CD45 OR ((TCR OR IL2RB) AND NOT ZAP70)) AND NOT Apoptosis \\
$f_{ZAP70}$ = LCK AND NOT (FYN OR Apoptosis) \\
$f_{GRB2}$ = (IL2RB OR ZAP70) AND NOT Apoptosis \\
$f_{PLCG1}$ = (GRB2 OR PDGFR) AND NOT Apoptosis \\
$f_{RAS}$ = (GRB2 OR PLCG1) AND NOT (GAP OR Apoptosis) \\
$f_{GAP}$ = (RAS OR (PDGFR AND GAP)) AND NOT (IL15 OR IL2 OR Apoptosis) \\
$f_{MEK}$ = RAS AND NOT Apoptosis \\
$f_{ERK}$ = (MEK AND PI3K) AND NOT Apoptosis \\
$f_{PI3K}$ = (PDGFR OR RAS) AND NOT Apoptosis \\
$f_{NFKB}$ = ((TPL2 OR PI3K) OR (FLIP AND TRADD AND IAP)) AND NOT Apoptosis \\
$f_{NFAT}$ = PI3K AND NOT Apoptosis \\
$f_{RANTES}$ = NFKB AND NOT Apoptosis \\
$f_{IL2}$ = (NFKB OR STAT3 OR NFAT) AND NOT (TBET OR Apoptosis) \\
$f_{IL2RBT}$ = (ERK AND TBET) AND NOT Apoptosis \\
$f_{IL2RB}$ = (IL2RBT AND (IL2 OR IL15)) AND NOT Apoptosis \\
$f_{IL2RAT}$ = (IL2 AND (STAT3 OR NFKB)) AND NOT Apoptosis \\
$f_{IL2RA}$ = (IL2 AND IL2RAT) AND NOT (IL2RA OR Apoptosis) \\
$f_{JAK}$ = (IL2RA OR IL2RB OR RANTES OR IFNG) AND NOT (SOCS OR CD45 OR Apoptosis) \\
$f_{SOCS}$ = JAK AND NOT (IL2 OR IL15 OR Apoptosis) \\
$f_{STAT3}$ = JAK AND NOT Apoptosis \\
$f_{P27}$ = STAT3 AND NOT Apoptosis \\
$f_{Proliferation}$ = STAT3 AND NOT (P27 OR Apoptosis) \\
$f_{TBET}$ = (JAK OR TBET) AND NOT Apoptosis \\
$f_{CREB}$ = (ERK AND IFNG) AND NOT Apoptosis \\
$f_{IFNGT}$ = (TBET OR STAT3 OR NFAT) AND NOT Apoptosis \\
$f_{IFNG}$ = ((IL2 OR IL15 OR Stimuli) AND IFNGT) AND NOT (SMAD OR P2 OR Apoptosis) \\
$f_{P2}$ = (IFNG OR P2) AND NOT (Stimuli2 OR Apoptosis) \\
$f_{GZMB}$ = ((CREB AND IFNG) OR TBET) AND NOT Apoptosis \\
$f_{TPL2}$ = (TAX OR (PI3K AND TNF)) AND NOT Apoptosis \\
$f_{TNF}$ = NFKB AND NOT Apoptosis \\
$f_{TRADD}$ = TNF AND NOT (IAP OR A20 OR Apoptosis) \\
$f_{FasL}$ = (STAT3 OR NFKB OR NFAT OR ERK) AND NOT Apoptosis \\
$f_{FasT}$ = NFKB AND NOT Apoptosis \\
$f_{Fas}$ = (FasT AND FasL) AND NOT (sFas OR Apoptosis) \\
$f_{sFas}$ = FasT AND S1P AND NOT Apoptosis \\
$f_{Ceramide}$ = Fas AND NOT (S1P OR Apoptosis) \\
$f_{DISC}$ = (FasT AND ((Fas AND IL2) OR Ceramide OR (Fas AND NOT FLIP))) AND NOT Apoptosis \\
$f_{Caspase}$ = ((((TRADD OR GZMB) AND BID) AND NOT IAP) OR DISC) AND NOT Apoptosis \\
$f_{FLIP}$ = (NFKB OR (CREB AND IFNG)) AND NOT (DISC OR Apoptosis) \\
$f_{A20}$ = NFKB AND NOT Apoptosis \\
$f_{BID}$ = (Caspase OR GZMB) AND NOT (BclxL OR MCL1 OR Apoptosis) \\
$f_{IAP}$ = NFKB AND NOT (BID OR Apoptosis) \\
$f_{BclxL}$ = (NFKB OR STAT3) AND NOT (BID OR GZMB OR DISC OR Apoptosis) \\
$f_{MCL1}$ = (IL2RB AND STAT3 AND NFKB AND PI3K) AND NOT (DISC OR Apoptosis) \\
$f_{Apoptosis}$ = Caspase OR Apoptosis \\
$f_{GPCR}$ = S1P AND NOT Apoptosis \\
$f_{SMAD}$ = GPCR AND NOT Apoptosis \\
$f_{SPHK1}$ = PDGFR AND NOT Apoptosis \\
$f_{S1P}$ = SPHK1 AND NOT (Ceramide OR Apoptosis) \\

\subsection*{Appendix D - The attractors of T-LGL leukemia survival network} \label{AppendixD}

In Table \ref{tab:Attractortable} we show the state of the nodes for all possible combinations of input signals in the presence of antigen (Stimuli=ON). We do not show the Apoptosis=ON attractor, in which all nodes except Apoptosis are inactive, since it is always a possibility. For simplicity, we only show which nodes oscillate and which of them stabilize in a steady state (i.e., the quasi-attractor)
and not the actual attractor, which would include all the network states that the nodes that oscillate can visit along with the transitions between these states. The signal combinations not shown in the table (i.e. CD45=OFF and IL15=OFF, with any other value for the other input signals) have apoptosis as their only attractor.

\begin{table*}
    \rowcolors{5}{white}{lightgray}
    \scalebox{0.7}{
    \begin{tabular}{ || >{\centering}p{2.5cm}<{\centering} || >{\centering}p{3.5cm}<{\centering} || >{\centering}p{3.5cm}<{\centering} || >{\centering}p{3.5cm}<{\centering} || p{3.5cm}<{\centering} ||}
    \hline
    \hline
    \ & \textbf{CD45=ON} & \textbf{CD45=ON} & \textbf{CD45=OFF} & \textbf{CD45=OFF}\\
    \ & \textbf{PDGF=ON/OFF} & \textbf{PDGF=ON/OFF} & \textbf{PDGF=ON/OFF} & \textbf{PDGF=ON/OFF}\\
    \textbf{Node} & \textbf{IL15=ON/OFF} & \textbf{IL15=ON/OFF} & \textbf{IL15=ON} & \textbf{IL15=ON}\\
    \ & \textbf{Stimuli2=OFF} & \textbf{Stimuli2=ON} & \textbf{Stimuli2=OFF} & \textbf{Stimuli2=ON}\\
    \ & \textbf{TAX=ON/OFF} & \textbf{TAX=ON/OFF} & \textbf{TAX=ON/OFF} & \textbf{TAX=ON/OFF}\\
    \hline \hline
    IL2RBT & OFF & OFF & ON & ON \\ \hline
    BclxL & ON & ON & OFF & OFF \\ \hline
    IFNGT & ON & ON & ON & ON \\ \hline
    PDGFR & ON & ON & ON & ON \\ \hline
    IFNG & OFF & OFF & OFF & OFF \\ \hline
    GAP & OFF & OFF & OFF & OFF \\ \hline
    Proliferation & OFF & OFF & OFF & OFF \\ \hline
    GZMB & OFF & OFF & ON & ON \\ \hline
    RAS & ON & ON & ON & ON \\ \hline
    TPL2 & ON & ON & ON & ON \\ \hline
    FasT & ON & ON & ON & ON \\ \hline
    FLIP & ON & ON & ON & ON \\ \hline
    LCK & ON & ON & ON & ON \\ \hline
    NFAT & ON & ON & ON & ON \\ \hline
    FasL & ON & ON & ON & ON \\ \hline
    Caspase & OFF & OFF & OFF & OFF \\ \hline
    NFKB & ON & ON & ON & ON \\ \hline
    IAP & ON & ON & ON & ON  \\ \hline
    BID & OFF & OFF & OFF & OFF \\ \hline
    Cyto. Signal. & Oscillates & Oscillates & ON & ON\\ \hline
    TNF & ON & ON & ON & ON \\ \hline
    MCL1 & OFF & OFF & ON & ON\\ \hline
    Ceramide & OFF & OFF & OFF & OFF \\ \hline
    GRB2 & Oscillates & Oscillates & ON & ON \\ \hline
    PI3K & ON & ON & ON & ON \\ \hline
    SMAD & ON & ON & ON & ON \\ \hline
    P27 & OFF & OFF & ON & ON \\ \hline
    ZAP70 & Oscillates & Oscillates & OFF & OFF \\ \hline
    CREB & OFF & OFF & OFF & OFF \\ \hline
    DISC & OFF & OFF & OFF & OFF \\ \hline
    IL2RB & OFF & OFF & ON & ON\\ \hline
    Fas & OFF & OFF & OFF & OFF \\ \hline
    IL2RA & Oscillates & Oscillates & OFF & OFF \\ \hline
    S1P & ON & ON & ON & ON \\ \hline
    ERK & ON & ON & ON & ON \\ \hline
    SPHK1 & ON & ON & ON & ON \\ \hline
    A20 & ON & ON & ON & ON \\ \hline
    MEK & ON & ON & ON & ON \\ \hline
    CTLA4 & Oscillates & Oscillates & Oscillates & Oscillates\\ \hline
    TBET & OFF & OFF & ON & ON \\ \hline
    RANTES & ON & ON & ON & ON \\ \hline
    SOCS & OFF & OFF & OFF & OFF \\ \hline
    sFas & ON & ON & ON & ON \\ \hline
    IL2RAT & ON & ON & OFF & OFF\\ \hline
    TCR & Oscillates & Oscillates & Oscillates & Oscillates\\ \hline
    STAT3 & OFF & OFF & ON & ON \\ \hline
    GPCR & ON & ON & ON & ON \\ \hline
    P2 & ON, OFF & OFF & ON, OFF & OFF \\ \hline
    TRADD & OFF & OFF & OFF & OFF \\ \hline
    PLCG1 & ON & ON & ON & ON \\ \hline
    FYN & Oscillates & Oscillates & ON & ON \\ \hline
    IL2 & ON & ON & OFF & OFF \\ \hline
    JAK & OFF & OFF & ON & ON \\ \hline
    Apoptosis & OFF & OFF & OFF & OFF \\ \hline
    \end{tabular}
    }
\caption{The attractors of T-LGL leukemia survival network. This table shows the state of the nodes for all possible combinations of input signals in the presence of antigen (Stimuli=ON).}
\label{tab:Attractortable}
\end{table*}

\end{document}